\documentclass[reqno,12pt]{amsart}

\usepackage{color} % black,white,red,green,blue,cyan,magenta,yellow
\usepackage[dvipsnames]{xcolor}
\usepackage[pdftex]{graphicx}
\usepackage[pdftex]{hyperref}
\usepackage{hyperref}
\hypersetup{
    unicode=false,          % non-Latin characters in Acrobat’s bookmarks
    pdftoolbar=true,        % show Acrobat’s toolbar?
    pdfmenubar=true,        % show Acrobat’s menu?
    pdffitwindow=false,     % window fit to page when opened
    pdfstartview={FitH},    % fits the width of the page to the window
    pdftitle={Article},     % title
    pdfauthor={Collaboration},   % author
    pdfsubject={Mathematics},    % subject of the document
    pdfcreator={PI},             % creator of the document
    pdfproducer={PI},            % producer of the document
    pdfkeywords={PDE, analysis}, % list of keywords
    pdfnewwindow=true,          % links in new window
    colorlinks=true,            % false: boxed links; true: colored links
    linkcolor=black,            % color of internal links
    citecolor=blue,             % color of links to bibliography
    filecolor=magenta,          % color of file links
    urlcolor=cyan               % color of external links
}

\usepackage[color=green!20]{todonotes}

\usepackage{adjustbox} \usepackage[version=4]{mhchem}

\usepackage{parskip}
\usepackage{times}
\usepackage{amsfonts}
\usepackage{amsmath}
\usepackage{amsthm}
\usepackage{amssymb}
\usepackage{bbm}
\usepackage{amsbsy}
\usepackage{amscd}
\usepackage{enumerate}
\usepackage{verbatim}
\usepackage{subfigure}
\usepackage{cleveref}
\usepackage{changepage}
\usepackage[algo2e,linesnumbered,ruled]{algorithm2e}

\newtheorem{theorem}{Theorem}[section]

\newtheorem{proposition}[theorem]{Proposition}

\numberwithin{equation}{section}  %amsmath command: tie counter to section

\usepackage[english]{babel}
\usepackage{chemformula}
{}  % helvetica font
\newtheorem{thm}{Theorem}[section]

\begin{document}

\vspace*{-0.9cm}

\title
[Stability analysis of a circuit with dual GTPase switches]
{Stability analysis of a signaling circuit with dual species of GTPase switches}

\author[L.M. Stolerman]{Lucas M. Stolerman$^{1}$}
\author[P. Ghosh]{Pradipta Ghosh$^{2,3,4*}$}
\author[P. Rangamani]{Padmini Rangamani$^{1*}$}

\thanks{$^1$ \scriptsize  Department of Mechanical and Aerospace Engineering, University of California, San Diego, La Jolla CA 92093}
\thanks{$^2$  Department of Medicine,University of California, San Diego, La Jolla, CA 92093}
\thanks{$^3$Department of Cellular and Molecular Medicine, University of California, San Diego, La Jolla, CA 92093}
\thanks{$^4$ Moores Comprehensive Cancer Center, University of California, San Diego, La Jolla, CA 92093}
\thanks{$^{*} $To whom correspondence should be addressed. e-mail: prangamani@ucsd.edu; prghosh@health.ucsd.edu}

\date{\today}
%\maketitle
%\tableofcontents


\normalsize

\begin{abstract} 
GTPases are molecular switches that regulate a wide range of cellular processes, such as organelle biogenesis, position, shape, and function, vesicular transport between organelles, and signal transduction. 
These hydrolase enzymes operate by toggling between an active ("ON") guanosine triphosphate (GTP)-bound state and an inactive ("OFF") guanosine diphosphate (GDP)-bound state; such a toggle is regulated by GEFs (guanine nucleotide exchange factors) and GAPs (GTPase activating proteins). 
Here we dissect a network motif between monomeric (m) and trimeric (t) GTPases assembled exclusively in eukaryotic cells of multicellular organisms. 
To this end, we develop a system of  ordinary differential equations in which these two classes of GTPases are interlinked conditional to their ON/OFF states within a motif through feedforward and feedback loops. 
We provide explicit formulas for the steady states of the system and perform classical local stability analysis 
to systematically investigate the role of the different connections between the GTPase switches. 
Interestingly, a feedforward from the active mGTPase to the GEF of the tGTPase was sufficient to provide two locally stable states: one where both active/inactive forms of the mGTPase can be interpreted as having low concentrations and the other where both m- and tGTPase have high concentrations.
Moreover, when a feedback loop from the GEF of the tGTPase to the GAP of the mGTPase was added to the feedforward system, two other locally stable states emerged, both having the tGTPase inactivated and being interpreted as having low active tGTPase concentrations.
Finally, the addition of a second feedback loop, from the active tGTPase to the GAP of the mGTPase, gives rise to a family of steady states that can be parametrized by a range of inactive tGTPase concentrations.
Our findings reveal that the coupling of these two different GTPase motifs can dramatically change their steady state behaviors and shed light on how such coupling may impact information processing in eukaryotic cells.

\end{abstract}

\maketitle

\vspace*{-0.8cm}
{\scriptsize
\newpage
\tableofcontents
}
\vspace*{-1.2cm}

%\clearpage
%%%%%%%%%%%%%%%%%%%%%%%%%%%%%%%%%%%%%%%%%%%%%%%%%%%%%%%%%%%%%%%%%%%%%%%%%%%%%%

\newpage
\section{Introduction}
\noindent

 % Paragraph 1: Broad introduction to GTPases
 Each eukaryotic cell has a large number of GTP-binding proteins (also called  \emph{GTPases} or G-proteins); they are thought to be intermediates in an extended cellular signaling and transport network that touches on nearly every aspect of cell function  \cite{alberts2013molecular,hamm1998many,bourne1990gtpase}.  
 One unique feature of GTPases is that they serve as \emph{biochemical switches} that exist in an `OFF' state when bound to a guanosine diphosphate (GDP), and can be turned `ON' when that GDP is exchanged for a guanosine triphosphate (GTP) nucleotide.
\cite{alberts2013molecular,lipshtat2010design}. 
Turning the GTPase `ON' is the key rate limiting step in the activation-inactivation process, requires an external stimulus, and is catalyzed by a class of enzymes called guanine nucleotide exchange factors (GEFs) \cite{rossman2005gef}.
G proteins return to their `OFF' state when the bound GTP is hydrolyzed to guanosine diphosphate (GDP) via an intrinsic hydrolase activity of the GTPase; this step is catalyzed by GTPase-activating proteins (GAPs) \cite{wang1999gtpase}. 
Thus, GEFs and GAPs play a crucial role in controlling the dynamics  of the GTPase switch and the finiteness of signaling that it transduces \cite{bos2007gefs,cherfils2013regulation,siderovski2005gaps,ghosh2017gaps}. %\needref. 
Dysregulation of GTPase switches has been implicated in cellular malfunctioning and is commonly encountered in diverse diseases \cite{lopez2014giv,ma2015therapeutic,wang2015giv,hartung2013akt}.  
For example, hyperactivation of GTPases \cite{digiacomo2020probing,hanahan2000hallmarks} is known to support a myriad of cellular phenotypes that contribute to aggressive tumor traits  \cite{cardama2017rho,liu2017thirty}. 
Such traits have also been associated with aberrant activity of GAPs \cite{digiacomo2020probing} or GEFs \cite{sriram2019gpcrs,wu2019illuminating,o2013emerging,garcia2009giv,ghosh2015heterotrimeric,papasergi2015g}. 
These works underscore the importance of GTPases as vital regulators of high fidelity cellular communication.
 
There are two distinct types of GTPases that gate signals: small or monomeric (m) and heterotrimeric (t) GTPases. mGTPases are mostly believed to function within the cell's interior and are primarily concerned with organelle function and cytoskeletal remodeling \cite{evers2000rho,etienne2002rho,takai2001small}.
tGTPases, on the other hand, were believed to primarily function at the cell's surface from where they gate the duration, type and extent of signals that are initiated by receptors on the cell's surface \cite{gilman1987g,barr1992trimeric}. 
These two classes of switches were believed to function largely independently, until early 1990's when tGTPases were detected on intracellular membranes, e.g., the Golgi \cite{barr1992trimeric,stow1991heterotrimeric}, and studies alluded to the possibility that they, alongside mGTPases, may co-regulate organelle function and structure \cite{cancino2013signaling}. 
But it was not until 2016 that the first evidence of functional coupling between the two switches -- m- and tGTPases-- emerged \cite{lo2015activation}. 
Using a combination of biochemical, biophysical, structural modeling, live cell imaging, and numerous readouts of Golgi functions, it was shown that m- and tGTPase co-regulate each other on the Golgi \cite{lo2015activation}. 
mGTPase (Arf family) must be turned `ON' for it to engage with a GEF (GIV/Girdin) for tGTPase, and the latter subsequently triggers the activation of a tGTPase (of the Gi sub-family). 
Upon activation, the tGTPase enhances a key GAP for mGTPases (i.e., ArfGAP2/3) that turns `OFF' the mGTPase, thereby terminating mGTPase signaling.  
This phenomenon of co-regulation between the two classes of GTPases was shown to be critical for maintaining the finiteness of signaling via both species of GTPase on the membrane, which was critical for maintaining Golgi shape and function (trafficking within the secretory pathway). In doing so, this dual GTPase circuit was converting simple chemical signals into complex mechanical outputs (membrane trafficking). 
Because emerging evidence from protein-protein interaction networks and decades of work on both species of GTPases suggest that co-regulation through coupling between the GTPases is possible and likely on multiple organellar membranes, it begs the question -- what are the advantages of two species of GTPase switches coupled through a closed feedback loop within a network motif as opposed to single switches? 
The answer to this question has not yet been experimentally dissected or intuitively theorized.   

Mathematical models of signaling networks have contributed significantly to our understanding of how network motifs might function \cite{alon2019introduction,bower2001computational,eungdamrong2004modeling,cowan2012spatial,morris2010logic,getz2019predictive}. 
 Continuous-time dynamical systems, commonly represented by systems of ordinary differential equations (ODEs), are powerful tools for building rich and insightful mathematical models  \cite{milo2002network}. 
For example, a comprehensive steady state analysis of ODE system helped frame the concept of ``zero-order ultrasensitivity'' where large responses in the active fraction of a protein of interest are driven by small changes in the saturated ratio of the enzymes \cite{goldbeter1981amplified}. 
Similarly, modeling biochemical networks with dynamical systems also has revealed the existence of bistable switches and biological oscillators within a feedback network architecture \cite{ferrell2002self,ingolia2007positive,ferrell2001bistability}.  
The simple system proposed by Ferrell and Xiong \cite{ferrell2001bistability} served as a basis for modeling cellular all-or-none responses, and hence, crucial for decision-making within several signaling processes.
Dynamical systems have also be used for mapping chemical reactions into differential equations; when numerically integrated (clustered) these systems can be used to predict the evolution of large reaction networks \cite{alon2007network,neves2002,shen2002network}. 
For example, clustering methods have revealed the existence of recurrent structures, the so-called ``network motifs', the dynamics of which have been inferred with a Boolean kinetic system of differential equations \cite{shen2002network}. 
These models produced various dynamic features corresponding to the motif structure, which allowed the understanding of underlying biology (i.e., gene expression patters).
Furthermore, studies using network ODE models have revealed that information is processed in cells through intricate connections between signaling pathways rather than individual motifs \cite{bhalla1999a,bhalla1999emergent,bhalla2001robustness}. 
These findings led to the so-called ``learned behavior'' hypothesis, on comprised information within intracellular biochemical reactions, was verified by analyzing extended signal duration, feedback loops, and multiple thresholds of stimulation and signal outputs.
Finally, from a systems biology modeling perspective, large systems of differential equations are usually hard to analyze, but when combined with experiments, they can give rise to quantitative predictions \cite{logsdon2014systems,finley2009computational,finley2013effect}, as exemplified by the work of Yen et al. \cite{yen2011two}. Finally, dynamical systems modeling of biochemical networks can help understand how stochastic noise is attenuated or amplified within cells \cite{hornung2008noise,hooshangi2005ultrasensitivity,shibata2005noisy,pedraza2005noise, qiao2019network}, e.g., the trade-off between sensitivity and noise reduction in any given network topology, and how some network topologies could perform a \emph{dual function} of noise attenuation and adaptation \cite{qiao2019network}.

Here we build mathematical models to investigate the dynamic properties of a coupled biochemical GTPase switches. Beginning with the uncoupled GTPase switches (Fig.\ref{Fig1}A) as our starting point, we specifically sought to understand the stability features of a new network motif (Fig.\ref{Fig1}B) to gain insights into any relative advantages of coupling two distinct classes of GTPase switches. 
To this end, we obtained the steady states of this new network motif, to understand the input-output relationships.
Then we studied the dynamic behavior under small perturbation around these steady states using local stability analysis \cite{strogatz1994,perko2013differential}. We investigated the different feedforward and feedback loops between these two motifs, all representing the observed biochemical and biophysical events during signal transduction (Fig.\ref{Fig1}C). 
To gain biological intuition, we investigated the appropriate space of initial conditions that guarantee the existence of the different steady states, thus obtaining insight to the different regimes of the GTPase concentration levels.
In what follows, we present the model assumptions and derivation in \S \Cref{model_development}, the local stability analysis and numerical simulations in \S \Cref{mathematical_analysis},  and discuss our findings in the context of GTPase signalling networks in \S \Cref{discussion}. 

\begin{figure}
	\centering
	%\hspace{-3cm}
	\includegraphics[scale=1.6]{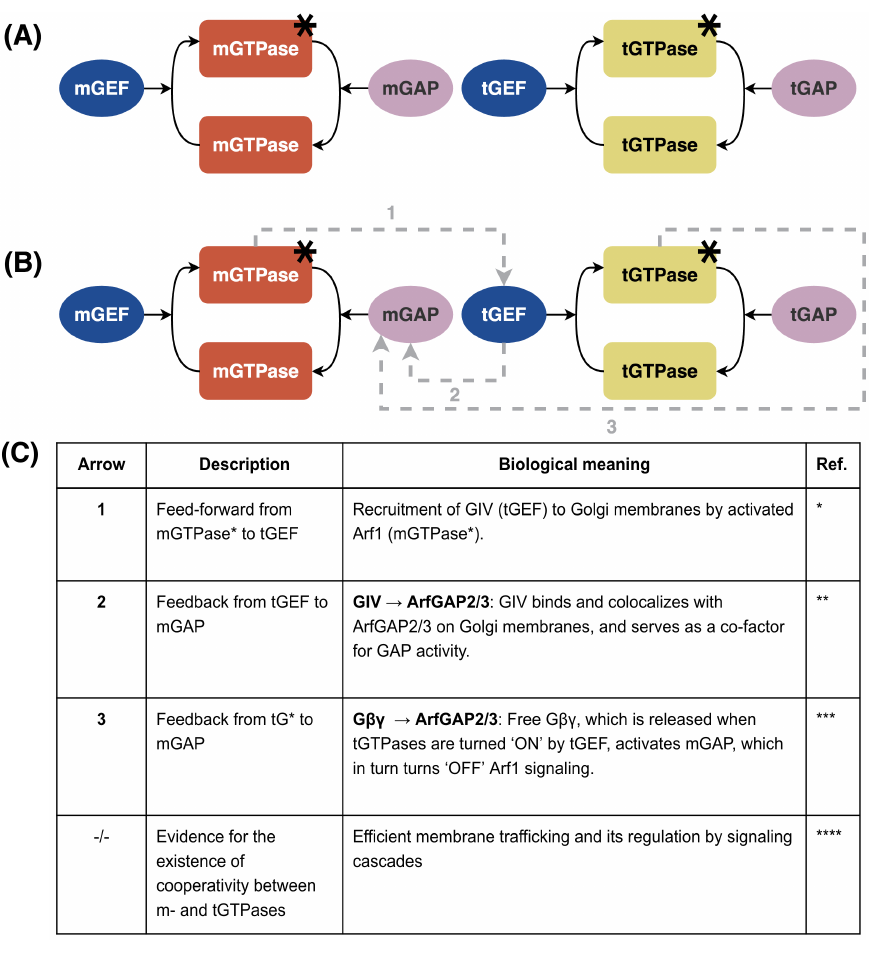}
	\caption{\textbf{A network motif in which two species of GTPases are interlinked.} 
	(A) Uncoupled monomeric and trimeric GTPase switches are represented by mGTPase and tGTPase, respectively. The black star denotes the active forms. Activation and inactivation are regulated by GEFs and GAPs, where the first letter (m or t) indicates the associated GTPase. (B) Our proposed mathematical model describes the interaction between the two GTPase switches.  Arrows 1, 2, and 3 show the feedforward and feedback loops that were found experimentally. (C) Description and biological meaning of each arrow connecting the GTPase switches. References:  \cite{lo2015activation} for arrow 1 (*), \cite{lo2015activation} for arrow 2 (**),  \cite{lo2015activation,jamora1997regulation} for arrow 3 (***), and \cite{jamora1997regulation,stow1991heterotrimeric,stow1998vesicle,stow1995regulation,cancino2013signaling} for evidence of cooperativity between m and tGTPases (****). }
 \label{Fig1}
\end{figure}

\section{Model Development}
\label{model_development}

In this section, we introduce our mathematical model for the GTPase coupled circuit  (Fig.\ref{Fig1}B). We begin with outlining the model assumptions in \S \Cref{Assumptions} and describe the reactions and governing equations in detail in \S \Cref{Governing_Equations}.

\subsection{Assumptions.}
\label{Assumptions}
Our model describes the time evolution of the concentrations for the different system components. 
\Cref{T:reactions}  contains the set of reactions in our system.  
In this work, we only consider the toggling of GTPases that are mediated by GEFs and GAPs that activate and inactivate them, respectively. 
Moreover, and guided directly by experimental findings \cite{lo2015activation}, we explicitly include the recruitment/engagement of tGEF by $\text{mG}^{*}$ (feedforward represented by arrow 1 in Fig.\ref{Fig1}B) and two feedback loops (arrows 2 and 3 in Fig.\ref{Fig1}B) representing the activation of mGAP by $\text{tGEF}^{*}$ and $\text{tG}^{*}$, respectively.  To develop the model equations, we consider a well-mixed regime to investigate how the kinetic reactions affects the dynamics of this GTPase circuit.  
We also assume that the concentrations of the species are in large enough amounts that deterministic kinetics hold \cite{gillespie2009deterministic,hahl2016comparison}. 
Finally, for mathematical tractability, all reactions in the system are governed by mass-action and, therefore, nonlinear kinetics such as Hill functions or Michaelis-Menten are not considered \cite{Changeux1967}.

\begin{table}[hbt!]
\caption{GTPase circuit reactions and rates used in the model}
\label{T:reactions}
\begin{center}
\begin{adjustbox}{width=\textwidth}
\begin{tabular}{clllll}
\hline
  & \textbf{List of Reactions} &  \textbf{Reaction Rate}\\
\hline
$\text{mG}^*$  activation &  \ce{ mG  + $\text{mGEF}^*$ ->C[$k^{mG}_{on}$][]   $\text{mG}^*$} &  $k^{mG}_{on}[mGEF^*][mG]$   \\
$\text{mG}^*$ inactivation & \ce{ $\text{mG}^*$  + $\text{mGAP}^*$ ->C[$k^{mG}_{off}$][] mG} &  $k^{mG}_{off}[mGAP^*][mG^*]$    \\
&&\\
Feedforward from mG* to tGEF (arrow 1) & \ce{$\text{mG}^*$  + $\text{tGEF}$ ->C[$k^{I}_{on}$][]   $\text{tGEF}^*$} &  $k^{I}_{on}[tGEF][mG^*]$   \\
&&\\
$\text{tG}^*$ activation & \ce{tG  + $\text{tGEF}^*$ ->C[$k^{tG}_{on}$][]   $\text{tG}^*$} &  $k^{tG}_{on}[tGEF^*][tG]$  \\

$\text{tG}^*$ inactivation & \ce{$\text{tG}^*$  + $\text{tGAP}^*$ ->C[$k^{tG}_{off}$][]  tG} &  $k^{tG}_{off}[tGAP^*][tG^*]$  \\
&&\\
Feedback loop from tGEF to mGAP (arrow 2) & \ce{$\text{tGEF}^*$  + $\text{mGAP}$ ->C[$k^{II}_{on}$][]   $\text{mGAP}^*$} & $k^{II}_{on} [tGEF^*][mGAP]$   \\
&&\\
Feedback loop tG* to mGAP (arrow 3) & \ce{$\text{tG}^*$  + $\text{mGAP}$ ->C[$k^{III}_{on}$][]   $\text{mGAP}^*$} &  $k^{III}_{on}[mGAP][tG^*]$  \\
\\ \hline
\end{tabular}
\end{adjustbox}
\end{center}
\end{table}

\subsection{Governing Equations.}
\label{Governing_Equations}

We developed a system of ODEs that describe coupled toggling of two switches, i.e., cyclical activation and inactivation of monomeric and trimeric GTPases within the network motif shown in Fig.\ref{Fig1}B and described in Fig.\ref{Fig1}C. In what follows, the brackets represent concentrations, which are nonnegative real numbers. The system equations are given by 

\small
\begin{eqnarray}
	\label{model_mod}
	%\left \{ 
	%\begin{array}{l}
	\frac{d[mG]}{dt}&=& -  k^{mG}_{on}[mGEF^*][mG] + k^{mG}_{off}[mGAP^*][mG^*] \label{model_mod:1}\\
	\frac{d[mG^*]}{dt}&=&  k^{mG}_{on}[mGEF^*][mG] - k^{mG}_{off}[mGAP^*][mG^*] - k^{I}_{on}[tGEF][mG^*] \label{model_mod:2}\\
	\frac{d[tG]}{dt}&=& -  k^{tG}_{on}[tGEF^*][tG] + k^{tG}_{off}[tGAP^*][tG^*]  \label{model_mod:3}\\
	\frac{d[tG^*]}{dt} &=&  k^{tG}_{on}[tGEF^*][tG]  -  k^{tG}_{off}[tGAP^*][tG^*] - k^{III}_{on}[mGAP][tG^*] \label{model_mod:4}\\
	\frac{d[tGEF]}{dt}&=& - k^{I}_{on}[tGEF][mG^*] \label{model_mod:5}\\
	 \frac{d[tGEF^*]}{dt}&=&  k^{I}_{on}[tGEF][mG^*] - k^{II}_{on} [tGEF^*][mGAP] \label{model_mod:6}\\
	 \frac{d[mGAP]}{dt}&=&  - k^{III}_{on}[mGAP][tG^*] - k^{II}_{on} [tGEF^*][mGAP]  \label{model_mod:7}\\
	  \frac{d[mGAP^*]}{dt}&=&  k^{III}_{on}[mGAP][tG^*] + k^{II}_{on} [mGAP] [tGEF^*] \label{model_mod:8}
	%\end{array} \right.
\end{eqnarray}
\normalsize

where $k^{mG}_{on}$ and $k^{tG}_{on}$  are the activation rates of mG and tG, respectively. The inactivation rates for mG and tG are given by $k^{mG}_{off}$ and $k^{tG}_{off}$. 
We denote $k^{I}_{on}$ as the rate of tGEF activation by $mG^*$ through the feedforward connection. 
Finally, the  mGAP activation rates through the $\text{tGEF}^*$ and $\text{tG}^*$ feedback loops are given by $k^{II}_{on}$ and $k^{III}_{on}$, respectively. 
To complete the system definition, all model components must have nonnegative initial conditions. 
In particular, if $k^{I}_{on} = k^{II}_{on} = k^{III}_{on} =0$, then our system describes two uncoupled GTPase switches (Fig.\ref{Fig1}A) that each have the same dynamics of the single GTPase model proposed in \cite{lipshtat2010design}.

\section{Mathematical analysis and results}
\label{mathematical_analysis}

In this section, we explore the role of the feedforward and feedback loops on the system dynamics.
To simplify our mathematical analysis and investigate the isolated  contributions of the different network architectures, we assume that all activation rates, when non-zero, are equal and given by $k_{on}$, as well the inactivation rates for both m- and t-GTPase switches  ($k^{mG}_{off} = k^{tG}_{off}=k_{off}$).
We also assume that the concentrations of $[mGEF^*]$ and $[tGAP^*]$ are constant in our model.

It is convenient to rewrite our ODE system in the form $$\frac{d\textbf{x}}{dt} = S. \textbf{v(x)},$$ where $\textbf{x}$ represents the vector of concentrations for the different components, $S$ is the stoichiometric matrix and $\textbf{v(x)}$ is a vector with the different reaction rates \cite{famili2003convex,rangamani2007survival}. 
Thus we define the components $ \textbf{x}(1) = [mG],  \textbf{x}(2) = [mG^*], \textbf{x}(3)= [tG], \textbf{x}(4) = [tG^*], \textbf{x}(5)= [tGEF], \textbf{x}(6)=  [tGEF^*], \textbf{x}(7) = [mGAP]$, and  $\textbf{x}(8) = [mGAP^*]$.
We also write the reaction velocities as  
   \begin{eqnarray*}
   	&& v_1 = k_{on}[mGEF^*]  \textbf{x}(1) ,v_2 = k_{off} \textbf{x}(2) \textbf{x}(8) , v_3 = k_{on} \textbf{x}(2)\textbf{x}(5), v_4 = k_{on} \textbf{x}(3)\textbf{x}(6)\\
   	&&
   	v_5 =  k_{off}[tGAP^*]\textbf{x}(4), v_6 = k_{on} \textbf{x}(6)\textbf{x}(7),v_7 = k_{on}  \textbf{x}(4)\textbf{x}(7).
   \end{eqnarray*}

The $8\times 7$ stoichiometric matrix for the system given by Eqs. \ref{model_mod:1} -- \ref{model_mod:8} is then  given by

\begin{equation}
S = 
\begin{bmatrix} 
&&&\text{\footnotesize{arrow 1}}&&&\text{\footnotesize{arrow 2}} & \text{\footnotesize{arrow 3}}\\
\text{mG}&-  1  & 1 & 0 & 0 &0 &0 &0   \\
\text{mG*}&1   &  - 1 & -1 & 0 &0 &0 &0   \\
\text{tG}&0  & 0 & 0 & -1 &1 &0 &0  \\
\text{tG*}&0  & 0 & 0 & 1 &-1 &0 &-1 \\
\text{tGEF}&0 & 0 & -1 & 0 &0 &0 &0   \\
\text{tGEF*}&0  & 0 & 1 & 0 &0 &-1 &0   \\
\text{mGAP}&0  & 0& 0 & 0 &0 &-1 &-1  \\
\text{mGAP*}&0  & 0 & 0 & 0 &0 &1 &1 
\end{bmatrix}
\label{stoich_matrix}
\end{equation}  

where the rows and columns of $S$ (Eq. \ref{stoich_matrix}) represent the $8$ components and $7$ reactions, respectively. The right null space of $ S $ comprises the steady state flux solutions, and the left null space contains the conservation laws of the system \cite{famili2003convex}. 
On the other hand, the column space contains the dynamics of the time-derivatives, and the rank of S is the actual dimension of the system in which the dynamics take place. 
In the following subsections, we analyze Eqs. \ref{model_mod:1} -- \ref{model_mod:8} when the mGTPase and tGTPase switches are coupled through: (i) A  feedforward connection $mG^* \to tGEF$ only (arrow 1), (ii) feedforward connection $mG^* \to tGEF$ and  feedback loop $tGEF \to mGAP$ (arrows 1 and 2) and (iii) feedforward connection $mG^* \to tGEF$ and feedback loops $tGEF \to mGAP$ and $tG^* \to mGAP$ (arrows 1, 2, and 3). 

\subsection{Feedforward connection: Recruitment of tGEF by active mGTPases $(mG^* \to tGEF)$.}

\label{subsec_Arrow1}

To analyze Eqs. \ref{model_mod:1} -- \ref{model_mod:8} with the feedforward connection only (arrow 1 in Fig.\ref{Fig1}B), we assume $k^{II}_{on} = k^{III}_{on} = 0$, which means that the feedback loop arrows 2 and 3 are not considered this first analysis. This represents the simple feedforward connection of the two GTPase switches, which, in cells appears to be mediated via activation--dependent coupling of mG* to tGEF \cite{lo2015activation}. In this case, the stoichiometric matrix (Eq. \ref{stoich_matrix}) is $8 \times 5$.

\subsubsection*{Conservation laws.}

For this particular system, the total concentrations $[tG_{tot}] := [tG] + [tG^*]$ and $[tGEF_{tot}] := [tGEF] + [tGEF^*]$  are constant over time and are strictly positive. 
For this reason,  it is convenient to introduce the fractions  $\mathcal{T} := \frac{[tG]}{[tG_{tot}]}$ and $ \mathcal{G} := \frac{[tGEF]}{[tGEF_{tot}]} $ of inactive tGTPase and tGEF in the system, respectively, and let $\mathcal{T}^*$ and $\mathcal{G}^*$ denote the fraction of their active forms. 
We then use $\mathcal{T}$ + $\mathcal{T}^* =1$ and $\mathcal{G}$ + $\mathcal{G}^* =1$   to rewrite the system in the form

\footnotesize
\begin{eqnarray}
	\frac{d[mG]}{dt}&=& -  k_{on} [mGEF^*][mG] + k_{off}[mGAP^*][mG^*]  \label{model_red_3:1}  \\
	\frac{d[mG^*]}{dt}&=&   k_{on} [mGEF^*][mG]  - k_{off}[mGAP^*][mG^*] -k_{on}[tGEF_{tot}](1 - \mathcal{G}^*)[mG^*] \label{model_red_3:2}  \\
	\frac{d\mathcal{T}^*}{dt}&=& k_{on} [tGEF_{tot}]  \mathcal{G}^* (1 - \mathcal{T}^*) - k_{off}[tGAP^*]\mathcal{T}^*  \label{model_red_3:3}   \\
	\frac{d\mathcal{G}^*}{dt}&=& k_{on}(1 - \mathcal{G}^*)[mG^*] \label{model_red_3:4} 
\end{eqnarray}
\normalsize

From the stoichiometric matrix (Eq. \ref{stoich_matrix}), we observe that
\begin{equation}
\label{constrain}
   [mG] +  [mG^*] + [tGEF_{tot}] \mathcal{G^*} = C
 \end{equation}
 where $C>0$ is constant over time. We compute the left null space of the stoichiometric matrix and confirm the total of three conservation laws in this case.
 The conservation law given by Eq. \ref{constrain} reduces the system to three unknowns, which eases the steady state and stability analysis.
 
 \subsubsection*{Steady states.} To find biologically plausible (nonnegative) steady states of the system given by Eqs. \ref{model_red_3:1} -- \ref{constrain}, we must find $\widehat{[mG]}$, $\widehat{[mG^*]}$, $\widehat{\mathcal{G}^*}$, and $\widehat{\mathcal{T}^*}$ such that the time-derivatives in Eqs. \ref{model_red_3:1}--\ref{model_red_3:4} are zero and the conservation law given by Eq. \ref{constrain} is satisfied. Defining $k = \frac{k_{off}}{k_{on}}$, we must solve the following system:

\small
\begin{eqnarray*}
 &&   [mGEF^*]\widehat{[mG]}  - k[mGAP^*]\widehat{[mG^*]} - [tGEF_{tot}](1 - \widehat{\mathcal{G}^*})\widehat{[mG^*]} = 0  \\
&&  [tGEF_{tot}]  \widehat{\mathcal{G}^*} (1 - \widehat{\mathcal{T}^*}) - k[tGAP^*]\widehat{\mathcal{T}^*}  = 0   \\
 && (1 -\widehat{\mathcal{G}^*})\widehat{[mG^*]} = 0  \\
 && 
 \widehat{[mG]} +  \widehat{[mG^*]} + [tGEF_{tot}] \widehat{\mathcal{G^*}} = C 
\end{eqnarray*}
\normalsize

 From the third equation above, we must have $\widehat{[mG^*]}=0$ or $\widehat{\mathcal{G}^*} = 1$. Thus we divide the steady state analysis in these two cases and summarize our results in the following proposition, whose proof can be found in the appendix \ref{appendix_1}.

\begin{proposition}
\label{prop_arrow1}
Let $k = \frac{k_{off}}{k_{on}}$. The steady states 
$ \widehat{\bf{x}}= \left(\widehat{[mG]},\widehat{[mG^*]},\widehat{\mathcal{T}^*},\widehat{\mathcal{G}^*}  \right) $
of the system given by Eqs. \ref{model_red_3:1} -- \ref{constrain} are given by 

\begin{itemize}
\item \underline{Steady state 1:}
\begin{equation}
\widehat{\bf{x}} = \left(0, 0, \frac{1}{1+\frac{k [tGAP^*]}{C}}, \frac{C}{[tGEF_{tot}]}  \right)
\label{ss1}
\end{equation}

if and only if $C \leq [tGEF_{tot}]$ and
\item  \underline{Steady state 2:}
\begin{eqnarray}
\label{ss2}
\widehat{\bf{x}} &=& 
\left(\frac{k [mGAP^*]}{[mGEF^*]+k[mGAP^*]} \left( C- [tGEF_{tot}]\right), \right. \nonumber\\
&& \left. \frac{[mGEF^*]}{[mGEF^*]+k[mGAP^*]} \left( C- [tGEF_{tot}]\right), \right. \nonumber \\
&& \left. \frac{[tGEF_{tot}]}{[tGEF_{tot}] + k[tGAP^*]},1  \right).
\end{eqnarray}
if and only if $C \geq [tGEF_{tot}]$.
\end{itemize}
\end{proposition}

Given the explicit expressions for the steady states and the parameter range in which they exist, we perform a local stability analysis to determine if these states are stable or unstable under small perturbations.
 We adopt the classical linearization procedure based on the powerful Hartman-Grobman theorem \cite{strogatz1994,perko2013differential}. We show that the steady states are \emph{locally asymptotically stable}, which means that any trajectory will be attracted to the steady state provided the initial condition is sufficiently close. 
\color{black}

\subsubsection*{Local Stability Analysis.}

Using that $[mG] = C - [mG^*] + [tGEF_{tot}] \mathcal{G^*}$ (from Eq. \ref{constrain}), we obtain the following three-dimensional system:

\small
\begin{eqnarray*}
	\frac{d[mG^*]}{dt}&=&  f_1([mG^*],\mathcal{G}^*,\mathcal{T}^*)    \\
	\frac{d\mathcal{T}^*}{dt}&=& f_2([mG^*],\mathcal{G}^*,\mathcal{T}^*)   \\
	\frac{d\mathcal{G}^*}{dt}&=& f_3([mG^*],\mathcal{G}^*,\mathcal{T}^*) 
\end{eqnarray*}
\normalsize

where 
\begin{equation*}
\begin{aligned}
     f_1([mG^*],\mathcal{G}^*,\mathcal{T}^*) & = k_{on} [mGEF^*]\left( C - [mG^*] - [tGEF_{tot}] \right) \\ 
     &- k_{off}[mGAP^*][mG^*] -k_{on}[tGEF_{tot}](1 - \mathcal{G}^*)[mG^*],
\end{aligned}
\end{equation*}

$$ f_2([mG^*],\mathcal{G}^*,\mathcal{T}^*)  =  k_{on} [tGEF_{tot}]  \mathcal{G}^* (1 - \mathcal{T}^*) - k_{off}[tGAP^*]\mathcal{T}^*,$$
and 
$$ f_3([mG^*],\mathcal{G}^*,\mathcal{T}^*)  = k_{on}(1 - \mathcal{G}^*)[mG^*].$$

To perform the local stability analysis, we calculate the Jacobian matrix  evaluated at the steady state \begin{equation}
\mathcal{J} \left[\widehat{[mG^*]},\widehat{\mathcal{T}^*},\widehat{\mathcal{G}^*}\right] =  \left. \begin{bmatrix} 
 \frac{\partial f_1}{\partial [mG^*]}   & \frac{\partial f_1}{\partial \mathcal{T}^*}  & \frac{\partial f_1}{\partial \mathcal{G}^*}   \\
 \frac{\partial f_2}{\partial [mG^*]}   & \frac{\partial f_2}{\partial \mathcal{T}^*}  & \frac{\partial f_2}{\partial \mathcal{G}^*} \\
\frac{\partial f_3}{\partial [mG^*]}   & \frac{\partial f_3}{\partial \mathcal{T}^*}  & \frac{\partial f_3}{\partial \mathcal{G}^*}
 \end{bmatrix} \right\rvert_{ \left(\widehat{[mG^*]},\widehat{\mathcal{T}^*},\widehat{\mathcal{G}^*}\right)}
 \label{jac_matrix}
\end{equation}
and by showing that all its eigenvalues have a negative real part, we can prove that the steady state is LAS \cite{strogatz1994}, provided we further assume that the strict inequalities from \Cref{prop_arrow1} hold. This is the content of the following theorem.

\begin{thm}
Let $C$ be the conservation quantity from Eq. \ref{constrain}. Then, 
\begin{enumerate} 
	\item  If $C < [tGEF_{tot}]$, the steady state 1 (Eq. \ref{ss1}) is LAS. 
	\item If $C > [tGEF_{tot}]$, the steady state 2  (Eq. \ref{ss2}) is LAS. 
\end{enumerate}
\label{thm_arrow1}
\end{thm}

\begin{proof}	
\noindent
 To simplify our notation, we introduce $k_{mGEF} = k_{on} [mGEF^*]$, $k_{mGAP} = k_{off} [mGAP]$, and $k_{tGAP} = k_{off} [tGAP^*]$. All calculations were done with MATLAB's R2019b symbolic toolbox and we proceed with the analysis of each case separately. 
\begin{enumerate}

\item Suppose $C < [tGEF_{tot}]$. As we have seen in the previous section, in this case  the steady state is given by Eq. \ref{ss1}.  The Jacobian matrix (Eq. \ref{jac_matrix})  is given by 
 
 \small
 \begin{equation*}
 \begin{bmatrix} 
 k_{on}\left(C - [tGEF_{tot}]\right) - k_{mGEF} - k_{mGAP}  & 0 & -k_{mGEF}[tGEF_{tot}] \\[0.5cm]
 0  & - C k_{on} - k_{tGAP} &  \frac{k_{tGAP} k_{on} [tGEF_{tot}]}{C k_{on} + k_{tGAP}} \\[0.5cm]
 -k_{on} \left(\frac{C}{[tGEF_{tot}]} - 1\right )& 0 &  0
 \end{bmatrix}.
 \end{equation*}

\normalsize

The first eigenvalue in this case is given by $\lambda_1 = - C k_{on} - k_{tGAP}$ and the other two ($\lambda_2$ and $\lambda_3$) are such that  
$$\lambda_2 + \lambda_3 = k_{on} (C - [tGEF_{tot}] ) - k_{mGAP} - k_{mGEF} 
 < 0$$ and $$\lambda_2 \lambda_3 = -k_{on} k_{mGEF} (C - [tGEF_{tot}]) >  0 $$ 
 from which we conclude that $\lambda_2$ and $\lambda_3$ are both negative and therefore the steady state is LAS.
 
 \item 	Suppose now that $C > [tGEF_{tot}]$. Following our previous analysis,  the steady state is given by Eq. \ref{ss2}.  The Jacobian matrix in this case is  given by 

\small
 \begin{equation*}
 \begin{bmatrix} 
- k_{mGAP} - k_{mGEF}  & 0 &  k_{mGEF} [tGEF_{tot}] \left(  \frac{ k_{on} \left(C - tGEF_{tot}\right)}{k_{mGAP} + k_{mGEF}} -  1\right) \\[0.5cm]
 0  & - k_{tGAP} - k_{on} [tGEF_{tot}] &  \frac{k_{tGAP} k_{on}  [tGEF_{tot}]}{k_{tGAP} + k_{on} [tGEF_{tot}]} \\[0.5cm]
  0 &  0 & -\frac{k_{on} k_{mGEF} (C - [tGEF_{tot}])}{k_{mGAP} + k_{mGEF}}
 \end{bmatrix}
 \end{equation*}

\normalsize
and the eigenvalues are given by $ \lambda_1 =- k_{tGAP} - k_{on} [tGEF_{tot}]$ ,$\lambda_2 = - k_{mGEF} - k_{mGAP} $ 
and $\lambda_3 = -\frac{k_{on} k_{mGEF}(C - [tGEF_{tot}])}{k_{mGAP} + k_{mGEF}}$, which are all negative and this completes the proof.
 \end{enumerate} 
\end{proof}

%\newpage
\subsubsection*{Biological interpretation of the stability features of the feedforward connection $mG^* \to  tGEF$.}

The feedforward connection from mG* to tGEF allows for the emergence of the steady state 1 (Eq. \ref{ss1}) with zero mG and mG* values.
This zero concentration can be interpreted as a scenario in which  nearly all the available mG proteins are activated to mG*, and that nearly all the mG* species have successfully engaged with the available tGEFs, thereby maximally recruiting tGEF on the Golgi membranes. 
The inequality $C<[tGEF_{tot}]$ must hold for existence and local asymptotic stability to steady state 1. Recalling the definition of $\mathcal{G}^*$ and that Eq. \ref{constrain} holds for all times, including $t=0$, this relationship between  $C$ and $[tGEF_{tot}]$ can be rewritten as $[mG](0) +  [mG^*](0) <[tGEF](0)$
where $[tGEF](0) = [tGEF_{tot}] - [tGEF^*](0)$ is initial concentration of cytosolic tGEF that is yet to be recruited by mG* to the membranes.
Thus, if the total amount of mG protein is initially less than the concentration of tGEF in cells, Theorem \ref{thm_arrow1} ensures that the steady state 1 with no mG and mG* will emerge. Moreover, the reduced system (Eqs. \ref{model_red_3:1}--\ref{constrain}) will converge to the steady state 1, provided the initial and steady state concentration values are sufficiently close.
Similarly, the steady state 2 will exist when $[mG](0) +  [mG^*](0) > [tGEF](0)$. This can be interpreted as a scenario where the total amount of mG protein concentration is higher than the total concentration of tGEF in cells.  In this case, the reduced system will converge to steady state 2  where some distribution of  mG, mG*, tG, tG* are present (Eq. \ref{ss2}), given sufficiently close  initial and steady state concentration values.

Fig.\ref{Fig2} illustrates the two possible steady states (gray-colored  ``1'' in the $2 \times 2$ table) promoted by the feedforward connection. steady state 1 can be interpreted as a configuration where the copy numbers of both active and inactive mGTPase are low, while the tGTPase copy numbers remain high. On the other hand, in steady state 2 both m- and tGTPases have high copy numbers in both their active and inactive forms. 
Our results suggest that the feedforward from mG to tGEF, which initiates the coupling between the two G protein switches, can drive the system to two possible configurations depending on the cellular concentrations of total mG and tGEF. 
If the initial tGEF is larger than the total mG, the feedforward connection will result in a significant decrease of the total mG and result in the activation of a fraction of the tGEF ($\widehat{\mathcal{G}^*} = \frac{C}{[tGEF_{tot}]}$ in Eq. \ref{ss1}). 
On the contrary, if the initial tGEF is less than the total mG, then the available tGEF will be fully engaged ($\widehat{\mathcal{G}^*} = 1$ in Eq. \ref{ss2}) and there will be a residual mG concentration in the system. 
We conclude that the initial difference between the copy numbers of total mG and tGEF (a cytosolic protein that is recruited to the membrane by mG*) is the main factor that will determine the steady state of the coupled GTPase switches. 

\color{black}
%\newpage
\subsection{Feedforward connection $mG^* \to tGEF$ with feedback loop $tGEF \to mGAP$:  Recruitment of tGEF by active mGTPases and tGEF colocalization with mGAP.}
\label{subsec_Arrows12}

We analyze the case where the feedback loop $tGEF \to mGAP$  (arrow 2 in Fig.\ref{Fig1}B) is added to the system with the feedforward connection. In cells, this feedback loop represents a  tGEF* colocalization with mGAP on Golgi membranes that facilitates the recruitment of GAP proteins \cite{lo2015activation}. In this case, the model equations are given by the following system:

\small
\begin{eqnarray}
	%\label{model_mod}
	%\left \{ 
	%\begin{array}{l}
	\frac{d[mG]}{dt}&=& -  k_{on}[mGEF^*][mG] + k_{off}[mGAP^*][mG^*] \label{model_12:1}\\
	\frac{d[mG^*]}{dt}&=&  k_{on}[mGEF^*][mG] - k_{off}[mGAP^*][mG^*] - k_{on}[tGEF][mG^*] \label{model_12:2}\\
	\frac{d[tG]}{dt}&=& -  k_{on}[tGEF^*][tG] + k_{off}[tGAP^*][tG^*]  \label{model_12:3}\\
	\frac{d[tG^*]}{dt} &=&  k_{on}[tGEF^*][tG]  -  k_{off}[tGAP^*][tG^*]  \label{model_12:4}\\
	\frac{d[tGEF]}{dt}&=& - k_{on}[tGEF][mG^*] \label{model_12:5}\\
	 \frac{d[tGEF^*]}{dt}&=&  k_{on}[tGEF][mG^*] - k_{on} [tGEF^*][mGAP] \label{model_12:6}\\
	 \frac{d[mGAP]}{dt}&=&  - k_{on} [mGAP] [tGEF^*]   \label{model_12:7}\\
	  \frac{d[mGAP^*]}{dt}&=&  k_{on} [mGAP] [tGEF^*] \label{model_12:8}
	%\end{array} \right.
\end{eqnarray}

\normalsize
where we keep all activation and inactivation rates at the same value ($k_{on}$ and $k_{off}$ respectively), as done in \Cref{subsec_Arrow1}. 
As we did in the previous section, we first analyze the conservation laws of this particular system. In this case, the stoichiometric matrix (Eq. \ref{stoich_matrix}) is $8 \times 6$.

\subsubsection*{Conservation Laws.}

We begin by observing that the total amount of tGTPase is conserved in this system. 
Thus we may use the fraction $\mathcal{T}^*$ as in \Cref{subsec_Arrow1} and that is the first conservarion law. 
The total amount of mGAP is also conserved, as we sum  Eqs. \ref{model_12:7} and \ref{model_12:8}. 
We can then write 
\begin{equation}
 [mGAP]  = [mGAP_{tot}] - [mGAP^*]
\label{mGAP_cons}
\end{equation}

and substitute the above expression for $[mGAP]$ in Eqs. \ref{model_12:6} and \ref{model_12:8}.  We choose to keep the concentrations of mGAP as a variable for notational simplicity and do not define its fraction. Summing Eqs. \ref{model_12:1}, \ref{model_12:2}, \ref{model_12:6}, and \ref{model_12:8}, and integrating over time, we get 
\begin{equation}
    [mG] + [mG^*] + [tGEF^*] + [mGAP^*] = C_1
    \label{constrain1_12}
\end{equation}

where $C_1 \geq 0$ is constant over time. Moreover, Eqs. \ref{model_12:5}, \ref{model_12:6}, and \ref{model_12:8} when summed and integrated give 
\begin{equation}
    [tGEF] + [tGEF^*] + [mGAP^*] = C_2
\label{constrain2_12}
\end{equation}
for $C_2 \geq 0$ also constant. We compute the left null space of the stoichiometric matrix (Eq. \ref{stoich_matrix}) and confirm the total of four conservation laws in this case. 
The reduced system is given by the following equations:

\small
\begin{eqnarray}
	\frac{d[mG]}{dt}&=& -  k_{on}[mGEF^*][mG] + k_{off}[mGAP^*][mG^*] \label{model_red12:1}\\
	\frac{d[mG^*]}{dt}&=&  k_{on}[mGEF^*][mG] - k_{off}[mGAP^*][mG^*] - k_{on}[tGEF][mG^*] \label{model_red12:2}\\
	\frac{d\mathcal{T}^*}{dt} &=&  k_{on}[tGEF^*](1- \mathcal{T}^*)  -  k_{off}[tGAP^*]\mathcal{T}^*  \label{model_red12:3}\\
	\frac{d[tGEF]}{dt}&=& - k_{on}[tGEF][mG^*] \label{model_red12:4}\\
	 \frac{d[tGEF^*]}{dt}&=&  k_{on}[tGEF][mG^*] - k_{on} [tGEF^*]\left( [mGAP_{tot}] - [mGAP^*] \right) \label{model_red12:5}\\
	   \frac{d[mGAP^*]}{dt}&=&  k_{on}  \left( [mGAP_{tot}] - [mGAP^*] \right) [tGEF^*] \label{model_red12:6}
\end{eqnarray}
\normalsize

with the conservation laws given by Eqs. \ref{constrain1_12} and \ref{constrain2_12}. In what follows, we calculate the steady states of the system. 

\subsubsection*{Steady states and local stability analysis.} 

We begin by introducing $k = \frac{k_{off}}{k_{on}}$ to simplify our notation. 
To find the steady states,  we must find nonnegative solutions of the following system: 

\small
\begin{eqnarray}
	  - [mGEF^*]\widehat{[mG]} + k \widehat{[mGAP^*]}\widehat{[mG^*]} &=& 0 \label{ss_red12:1}\\
      \widehat{[tGEF^*]}(1- \widehat{\mathcal{T}^*})  -  k [tGAP^*] \widehat{\mathcal{T}^*}  &=& 0 \label{ss_red12:3}\\
	    \widehat{[tGEF]}\widehat{[mG^*]} &=& 0 \label{ss_red12:4}\\
	      \left( [mGAP_{tot}] - \widehat{[mGAP^*]} \right) \widehat{[tGEF^*]} &=& 0 \label{ss_red12:5} \\
	      \widehat{[mG]} + \widehat{[mG^*]} + \widehat{[tGEF^*]} + \widehat{[mGAP^*]}  &=& C_1 \label{ss_red12:6} \\
	      	      \widehat{[tGEF]} + \widehat{[tGEF^*]} + \widehat{[mGAP^*]}  &=& C_2 \label{ss_red12:7}
\end{eqnarray}
\normalsize

From Eq. \ref{ss_red12:4}, we must have  $ \widehat{[tGEF]}=0$ or $\widehat{[mG^*]}=0$. Moreover, from Eq. \ref{ss_red12:5},  $\widehat{[tGEF^*]} = 0$ or $\widehat{[mGAP^*]}=[mGAP_{tot}]$ and thus we have four possible combinations to analyze. 

We study each case separately and obtain the necessary and sufficient inequalities involving the parameters $C_1$, $C_2$, and $[mGAP_{tot}]$ that ensure the existence of each steady state.
As we did in the previous section, we also show that the steady states are LAS provided the strict inequalities are satisfied.
We summarize our analysis in the following theorem, whose proof can be found in Appendix \ref{appendix_12}.

\begin{theorem}
\label{prop_12}
The steady states $$\widehat{\bf{x}} = \left(\widehat{[mG]},\widehat{[mG^*]},\widehat{\mathcal{T}^*},\widehat{[tGEF]},\widehat{[tGEF^*]}, \widehat{[mGAP^*]} \right)$$ of the system given by Eqs.  \ref{constrain1_12} - \ref{model_red12:6} are given by
\begin{itemize}
%%%%%% SS1
    \item \underline{Steady state 1:}% $\left(\widehat{[tGEF]}=0 \quad \text{and} \quad \widehat{[tGEF^*]} = 0\right)$,
\footnotesize
\begin{equation}
\widehat{\bf{x}} =  \left(0,0,\frac{C_1 - [mGAP_{tot}]}{(C_1 - [mGAP_{tot}]) + k [tGAP^*]}, C_2 - C_1,C_1 - [mGAP_{tot}],[mGAP_{tot}] \right)
\label{ss1_12_new}
\end{equation} 
\normalsize
if and only if $C_2 \geq C_1$ and $C_1 \geq [mGAP_{tot}]$. The steady state is LAS if $C_2 > C_1$ and $C_1 > [mGAP_{tot}]$.
%%%%% SS2
\item \underline{Steady state 2:}% $\left(\widehat{[tGEF]}=0 \quad \text{and} \quad \widehat{[mGAP^*]}=[mGAP_{tot}]\right)$,
\footnotesize
\begin{eqnarray}
\widehat{\bf{x}} &=& \left(\frac{k [mGAP_{tot}]\left( C_1 - C_2 \right)}{[mGEF^*] + k [mGAP_{tot}]}, \frac{[mGEF^*] \left( C_1 - C_2 \right)}{[mGEF^*] + k [mGAP_{tot}] }, \right. \nonumber\\
&& \left. \frac{C_2 - [mGAP_{tot}]}{ \left(C_2 - [mGAP_{tot}]\right) + k [tGAP^*]} ,0,C_2 - [mGAP_{tot}], [mGAP_{tot}] \right)
\label{ss2_12}
\end{eqnarray} 
\normalsize
if and only if  $C_1 \geq C_2$ and $C_2 \geq [mGAP_{tot}]$. The steady state is LAS if  $C_1 > C_2$ and $C_2 > [mGAP_{tot}]$.
%%%%%% SS3 
\item \underline{Steady state 3:} %$\left(\widehat{[mG^*]}=0 \quad \text{and} \quad \widehat{[tGEF^*]}=0\right)$
\footnotesize
\begin{equation}
\widehat{\bf{x}} =  \left(0,0,0,C_2-C_1,0,C_1\right)
\label{ss3_12}
\end{equation} 
\normalsize
if and only if  $C_2 \geq C_1$ and $C_1 \leq [mGAP_{tot}]$. The steady state is LAS if  $C_2 > C_1$ and $C_1 < [mGAP_{tot}]$.
%%%%%% SS4 
\item \underline{Steady state 4:} %$\left(\widehat{[mG^*]}=0 \quad \text{and} \quad \widehat{[mGAP^*]}=[mGAP_{tot}]\right)$
\footnotesize
\begin{equation}
\widehat{\bf{x}} = \left(\frac{k C_2}{[mGEF^*]+k C_2} \left( C_1- C_2\right),\frac{[mGEF^*]}{[mGEF^*]+k C_2} \left( C_1- C_2\right), 0,0,0, C_2 \right)
\label{ss4_12_new}
\end{equation} 
\normalsize
if and only if $C_1 \geq C_2$ and  $C_2 \leq [mGAP_{tot}]$. The steady state is LAS if  $C_1 > C_2$ and $C_2 < [mGAP_{tot}]$.
\end{itemize}
\end{theorem}

%\newpage
\color{black}
\subsubsection*{Biological interpretation of the stability features of the feedforward connection $mG^* \to tGEF$  with single feedback loop $tGEF \to mGAP$.}

%%%%% PART 1: Description of the ss in comparison with Arrow 1 only
The feedforward connection $mG^* \to tGEF$  together with the feedback loop $tGEF \to mGAP$ (arrows 1 and 2 in Fig.\ref{Fig1}B, respectively) allows for the emergence of four steady states. steady states 1 and 2 (Eqs. \ref{ss1_12_new} and \ref{ss2_12}) are similar to the two steady states obtained in \Cref{subsec_Arrow1}, although with different concentration values. On the other hand, steady states 3 and 4 (Eqs. \ref{ss3_12} and Eqs. \ref{ss4_12_new}) newly emerge in the system, both with tG* attaining zero concentration.  
%%%%%%%%%% PART 2: Biological interpretation of the steady states.
This zero concentration can be interpreted as a scenario in which nearly all the available tGTPase is inactivated. 
%%%%% PART 3:  Discussion of the initial conditions and mathematical results
Recalling the definitions of $C_1$ and $C_2$ and the fact that Eqs. \ref{constrain1_12} and \ref{constrain2_12} hold at all times, including at $t=0$, we can write  $C_1 =  [mG](0) + [mG^*](0) + [tGEF^*](0) + [mGAP^*](0)$ and  $C_2  =  [tGEF](0) + [tGEF^*](0) + [mGAP^*](0)$. In this way, from the inequalities obtained in \Cref{prop_12} for $C_1$ and $C_2$, we obtain relationships among the initial conditions of the original system (Eqs. \ref{model_12:1} -- \ref{model_12:8}) that are associated with each one of the four steady states.

For the existence and local asymptotic stability of steady state 1 (Eq. \ref{ss1_12_new}), where mG and mG* have zero concentration values, the inequalities $C_2 > C_1$ and $C_1 >[mGAP_{tot}]$
must hold. The first inequality can be written as $[mG](0) +  [mG^*](0) <[tGEF](0)$, which was obtained in \Cref{subsec_Arrow1} as the existence condition for the steady state with no mG and mG* (Eq. \ref{ss1}). On the other hand, the inequality $C_1 > [mGAP_{tot}]$ can be written as $ [mG](0) + [mG^*](0) + [tGEF^*](0) > [mGAP](0)$, where $[mGAP](0)$ is the initial concentration of cytosolic mGAP that is yet to be recruited by tGEF* to the membranes. Therefore, two conditions guarantee the existence of steady state 1: (1) The total amount of mG protein must be initially less than the concentration of tGEF and (2) The sum of the  concentrations of total mG protein and tGEF* must be initially higher than the concentration of mGAP. If both conditions hold, then \Cref{prop_12} ensures that steady state 1 will emerge and the reduced system (Eqs. \ref{model_red12:1} -- \ref{model_red12:6} along with Eqs. \ref{constrain1_12} and \ref{constrain2_12}) will converge to the steady state 1, provided the initial and steady state concentration values are sufficiently close. 

A similar analysis holds for steady states 2, 3 and 4. For simplicity, we present the required initial conditions for each steady state without repeating the conclusions that follows from \Cref{prop_12}.  For steady state 2 ((Eq. \ref{ss2_12})), where mG, mG*, tG, and tG* are present, the inequalities  $C_1 > C_2$ and $C_2 >[mGAP_{tot}]$ become $[mG](0) +  [mG^*](0) >[tGEF](0)$ and $[tGEF](0) + [tGEF^*](0) > [mGAP](0)$, respectively. Hence, the total amount of mG protein must be initially higher than the concentration of tGEF and the total amount of tGEF must be initially higher than concentration of mGAP. For steady state 3 (Eq. \ref{ss3_12}), where mG and mG* have zero concentration values and the tGTPase is fully inactivated, the inequalities  $C_2 > C_1$ and $C_1 < [mGAP_{tot}]$ become $[mG](0) + [mG^*](0) <[tGEF](0)$ and $ [mG](0) + [mG^*](0) + [tGEF^*](0) < [mGAP](0)$, respectively. Hence, the total amount of mG protein must be initially less than the concentration of tGEF and the sum of the concentrations of total mG protein and tGEF* must be initially less than the concentration of mGAP. For steady state 4 (Eq. \ref{ss4_123_new}), where mG and mG* are present and tG* concentration is zero, the inequalities $C_1 > C_2$ and $C_2 < [mGAP_{tot}]$ become $[mG](0) +  [mG^*](0) >[tGEF](0)$ and $[tGEF](0) + [tGEF^*](0) < [mGAP](0)$, respectively. Hence, the total amount of mG protein must be initially higher than the concentration of tGEF and the total amount of tGEF must be initially less than the concentration of mGAP.

Fig.\ref{Fig2} illustrates the four possible steady states (gray-colored ``1+2'' in the $2 \times 2$ table) promoted by the feedforward connection $mG^* \to tGEF$ and the feedback loop $tGEF \to mGAP$. Steady states 1 and 2 have the same interpretation of the two steady states obtained in \Cref{subsec_Arrow1}. On the other hand, steady states 3 and 4 were obtained through the sole contribution of the feedback loop $tGEF \to mGAP$. These states share the common feature of having tGTPase fully inactivated. However,  steady state 3 can be interpreted as a configuration where the copy numbers of mG and mG* are low, while in steady state 4, these copy numbers are high. 

\subsection{Feedforward connection $mG^* \to tGEF$ with feedback loops $tGEF \to mGAP$ and $tG^* \to mGAP$ :  Recruitment of tGEF by active mGTPases, tGEF colocalization with mGAP, and activation of mGAP by active tGTPases.}
\label{subsec_Arrows123}

We analyze the case where the feedback loop $tG^* \to mGAP$ (arrow 3 in Fig.\ref{Fig1}B) is added to the system with the feedforward connection and feedback loop  $tGEF \to mGAP$. This connection represents the release of free $G\beta\gamma$ promoting mGAP activation. We analyze the full system given by  Eqs. \ref{model_mod:1} -- \ref{model_mod:8} in the case where where all rates are equal. The model equations are thus given by the following system:

\small
 	\begin{eqnarray}
 	%\label{model_mod}
 	%\left \{ 
 	%\begin{array}{l}
 	\frac{d[mG]}{dt}&=& -  k_{on}[mGEF^*][mG] + k_{off}[mGAP^*][mG^*] \label{model_123:1} \\
 	\frac{d[mG^*]}{dt}&=&k_{on}[mGEF^*][mG] - k_{off}[mGAP^*][mG^*] - k_{on}[tGEF][mG^*] \label{model_123:2} \\
 	\frac{d[tG]}{dt}&=& -  k_{on}[tGEF^*][tG] + k_{off}[tGAP^*][tG^*] \label{model_123:3} \\
 	\frac{d[tG^*]}{dt} &=&  k_{on}[tGEF^*][tG]  -  k_{off}[tGAP^*][tG^*] - k_{on}[mGAP][tG^*] \label{model_123:4} \\
 	\frac{d[tGEF]}{dt}&=& - k_{on}[tGEF][mG^*] \label{model_123:5} \\
 	\frac{d[tGEF^*]}{dt}&=&  k_{on}[tGEF][mG^*] - k_{on} [mGAP] [tGEF^*] \label{model_123:6}  \\
 	\frac{d[mGAP]}{dt}&=&  - k_{on}[mGAP][tG^*] - k_{on} [tGEF^*][mGAP] \label{model_123:7} \\
 	\frac{d[mGAP^*]}{dt}&=&  k_{on}[mGAP][tG^*] + k_{on} [tGEF^*][mGAP]. \label{model_123:8}
 	%\end{array} \right.
 	\end{eqnarray}
\normalsize

In what follows, we compute the conservation laws and four 1-parameter steady state families  for the system given by  Eqs. \ref{model_123:1} -- \ref{model_123:8}. 
We also obtain the necessary conditions for the conserved quantities that guarantee the existence of each steady state family.  

\subsubsection*{Conservation Laws.} 

As in \ref{subsec_Arrows12}, we also observe that the total amount of mGAP is constant over time, so  Eq. \ref{mGAP_cons} still holds. 
On the other hand, the total tGTPase follows a new conservation law that we derive here. Summming  Eqs. \ref{model_123:1} -- \ref{model_123:4}, \ref{model_123:6}, and \ref{model_123:8}, we have
\begin{equation}
 [mG] + [mG^*] +  [tG] + [tG^*] + [tGEF^*] + [mGAP^*] = \tilde{C_1}. 
    \label{constrain1_123}
\end{equation}
and summing Eqs. \ref{model_123:3} -- \ref{model_123:6} and  Eq. \ref{model_123:8} and integrating over time, we obtain 
\begin{equation}
   [tG] + [tG^*] + [tGEF] + [tGEF^*] + [mGAP^*] = \tilde{C_2} 
   \label{constrain2_123}
\end{equation}
where $\tilde{C_1}$ and $\tilde{C_2}$ must be nonnegative constants. We compute the left null space of the stoichiometric matrix (Eq. \ref{stoich_matrix}) and confirm the total of three conservation laws, which are given by  Eqs. \ref{mGAP_cons}, \ref{constrain1_123}, and  \ref{constrain2_123}. These equations reduce Eqs. \ref{model_123:1} -- \ref{model_123:8} to a five-dimensional system, whose steady states can be obtained.

\subsubsection*{Steady states.} 

 We  compute the steady states of the system when the time derivatives in Eqs. \ref{model_123:1} -- \ref{model_123:8} are equal to zero.
Denoting $k = \frac{k_{off}}{k_{on}}$ as in the previous section and removing the linearly dependent equations, the problem reduces to finding the nonnegative solutions of the following system:

\small
\begin{eqnarray}
	  - [mGEF^*]\widehat{[mG]} + k \widehat{[mGAP^*]}\widehat{[mG^*]} &=& 0 \label{ss_red123:1}\\
    -  \widehat{[tGEF^*]} \widehat{[tG]}  +  k [tGAP^*] \widehat{[tG^*]}  &=& 0 \label{ss_red123:2}\\
     \left( [mGAP_{tot}] - \widehat{[mGAP^*]} \right) \widehat{[tG^*]} &=& 0 \label{ss_red123:3} \\
	    \widehat{[tGEF]}\widehat{[mG^*]} &=& 0 \label{ss_red123:4}\\
	      \widehat{[tGEF^*]}  \left( [mGAP_{tot}] - \widehat{[mGAP^*]} \right) &=& 0 \label{ss_red123:5} 
\end{eqnarray}
\normalsize

along with the conservation laws given by Eqs. \ref{mGAP_cons}, \ref{constrain1_123}, and \ref{constrain2_123}. 

Eq. \ref{ss_red123:2} gives $tG^* = \frac{\widehat{[tGEF^*]} \widehat{[tG]}}{k [tGAP^*]}$ and Eq. \ref{ss_red123:3} then becomes
$$  \left( [mGAP_{tot}] - \widehat{[mGAP^*]} \right) \widehat{[tGEF^*]} \widehat{[tG]} = 0.$$
From Eq. \ref{ss_red123:5}, we conclude that $\widehat{[tG]}$ can be any nonnegative real number satisfying Eqs. \ref{constrain1_123} and \ref{constrain2_123}. We define  $\xi = \widehat{[tG]}$ and characterize four $\xi$-dependent families of steady states similarly as we did in \Cref{subsec_Arrows12}.  We summarize our results in the following theorem, whose proof can be found in the appendix \ref{appendix_123}.

\begin{theorem}
\label{prop_123}
The $\xi$-dependent families of steady states $$\widehat{\bf{x}}_{\xi} = \left(\widehat{[mG]},\widehat{[mG^*]},\widehat{[tG]},\widehat{[tG^*]},\widehat{[tGEF]},\widehat{[tGEF^*]}, \widehat{[mGAP^*]} \right)$$ of the system given by Eqs. \ref{ss_red123:1} -- \ref{ss_red123:5} with the conservation laws given by Eqs. \ref{mGAP_cons}, \ref{constrain1_123}, and \ref{constrain2_123} are given by
\begin{itemize}

%%%%%% SS4 --- new SS 1
\item \underline{Family 1:} %$\left(\widehat{[mG^*]}=0 \quad \text{and} \quad \widehat{[mGAP^*]}=[mGAP_{tot}]\right)$
\footnotesize
\begin{eqnarray}
\widehat{\textbf{x}}_{\xi}
&=& \left(0, 0, \xi, \frac{\left(\tilde{C_1} - [mGAP_{tot}] -\xi \right) \xi}{k[tGAP^*] +\xi}, \tilde{C_2} - \tilde{C_1}, \right. \nonumber \\
&& \left. \left(\tilde{C_1} - [mGAP_{tot}] -\xi \right) \frac{k [tGAP^*]}{k[tGAP^*] +\xi}, [mGAP_{tot}] \right.\Bigg)  %\nonumber
\label{ss1_123_new}
\end{eqnarray}
\normalsize
only if $0 \leq \xi + [mGAP_{tot}]\leq \tilde{C_1} \leq \tilde{C_2}$ .

%%%%% SS2
\item \underline{Family 2:}
\scriptsize
\begin{eqnarray}
\widehat{\textbf{x}}_{\xi}
&=& \left(\frac{(\tilde{C_1} - \tilde{C_2}) k[mGAP_{tot}] }{ [mGEF^*] + k[mGAP_{tot}] }, \frac{(\tilde{C_1} - \tilde{C_2}) [mGEF^*]}{ [mGEF^*] + k[mGAP_{tot}]},  \right. \nonumber \\
&& \left. \xi , \frac{\left(\tilde{C_2} - [mGAP_{tot}] -\xi] \right) \xi}{k[tGAP^*] +\xi} ,0, \left(\tilde{C_2} - [mGAP_{tot}] -\xi] \right) \frac{k [tGAP^*]}{k[tGAP^*] +\xi} , [mGAP_{tot}] \right.\Bigg)
\label{ss2_123}
\end{eqnarray}
\normalsize
 only if  $0 \leq \xi + [mGAP_{tot}]\leq \tilde{C_2} \leq \tilde{C_1}$.

%%%%%% SS3 
\item \underline{Family 3:} %$\left(\widehat{[mG^*]}=0 \quad \text{and} \quad \widehat{[tGEF^*]}=0\right)$
\footnotesize
\begin{equation}
\widehat{\bf{x}}_{\xi} =  \left(0,0,\xi,0,\tilde{C_2}-\tilde{C_1},0,\tilde{C_1} - \xi \right)
\label{ss3_123}
\end{equation} 
\normalsize
 only if  $\max(0,\tilde{C_1} - [mGAP_{tot}]) \leq \xi \leq \tilde{C_1} \leq \tilde{C_2}.$

%%%%%% SS1 --- new SS 4
    \item \underline{Family 4:}% $\left(\widehat{[tGEF]}=0 \quad \text{and} \quad \widehat{[tGEF^*]} = 0\right)$,
\footnotesize
\begin{eqnarray}
\widehat{\textbf{x}}_{\xi}
&=& \left(\frac{(\tilde{C_1} - \tilde{C_2}) k(\tilde{C_2} - \xi) }{ [mGEF^*] + k(\tilde{C_2} - \xi) }, \frac{(\tilde{C_1} - \tilde{C_2}) [mGEF^*]}{ [mGEF^*] + k(\tilde{C_2} - \xi)}, \xi ,0,0,0, \tilde{C_2}- \xi \right). 
\label{ss4_123_new}
\end{eqnarray} 
\normalsize
 only if $\max(0,\tilde{C_2} - [mGAP_{tot}]) \leq \xi \leq \tilde{C_2} \leq \tilde{C_1}$.
\end{itemize}
\end{theorem}

\subsubsection*{Biological interpretation of the stability features of the feedforward connection $mG^* \to tGEF$  with two feedback loops $tGEF \to mGAP$ and $tG^* \to mGAP$.}

\color{black}
The feedback loop $tG^* \to mGAP$, when added to the feedforward connection $mG^* \to tGEF$ and the feedback loop $tGEF \to mGAP$, allows for the emergence of four steady state families. Those families have the inactive tGTPase with different range values that can be obtained in the steady state analysis. Interestingly, the Families 1 -- 4 resemble the steady states 1 -- 4 from \Cref{subsec_Arrows12}. Family 1 has no mG and mG* at steady state, and both tG and tG* have non-zero steady state values (similarly to steady state 1). Moreover, Family 2 has both m- and tGTPases with non-zero steady states (similarly to steady state 2). For Family 3, mG and mG* steady state values are zero, and the tGTPase is fully inactivated (similarly to steady state 3). Finally, Family 4 has tGTPase is fully inactivated, and both mG and mG* have non-zero steady states (similarly to steady state 4). Recalling the definitions of $\tilde{C_1}$ and $\tilde{C_2}$ and the fact that  Eqs. \ref{constrain1_123} and \ref{constrain2_123} hold at all times, including $t=0$, we can infer necessary relationships among the initial conditions for each steady state family.

The inequality $\tilde{C_1} \leq  \tilde{C_2}$ can be rewritten as $[mG](0) +[mG^*](0) \leq [tGEF](0)$ is necessary for the emergence of Family 1 (Eq. \ref{ss1_123_new}) whith zero mG and mG* values, which can be interpreted as a scenario in which  nearly all the available mG proteins are activated to mG*, and that nearly all the mG* species have successfully engaged with the available tGEFs, thereby maximally recruiting tGEF on the membranes. For Family 1, $[mGAP_{tot}] \leq \tilde{C_1}$ also holds and can be written as $[mGAP](0) \leq [mG](0) + [mG^*](0) + [tG](0) + [tG^*](0) + [tGEF^*](0)$, where $[mGAP](0)$ is the initial concentrations of cytosolic mGAP that is yet to be recruited by tGEF* and tG* to the membranes. Therefore, two initial conditions are necessary for the existence of Family 1: The total amount of mG protein must be initially less than the concentration of tGEF and (2) The summed concentrations of total mG, total tG and tGEF* must be initially higher than the concentration of mGAP. Finally, the inequality $0 \leq \xi + [mGAP_{tot}] \leq \tilde{C}_1$ can be written as $0 \leq \xi \leq  [mG](0) + [mG^*](0) + [tG](0) + [tG^*](0) + [tGEF^*](0) - [mGAP](0)$. Remarkably, we conclude that the initial \emph{balance} between the summed concentrations of total mG, total tG, tGEF* and the available mGAP is the upper bound for the tG concentration, which completely characterizes the necessary conditions for the emergence of Family 1.

A similar analysis can be done for Families 2, 3, and 4. For the existence of Family 2 (Eq. \ref{ss2_123}), where mG, mG* tG, tG* are present (when $\xi>0$), the inequalities $\tilde{C_2} \leq  \tilde{C_1}$ and $[mGAP_{tot}] \leq \tilde{C_2}$ must hold and can be rewritten as $[mG](0) +[mG^*](0) \geq [tGEF](0)$ and $[mGAP](0) \leq [tG](0) + [tG^*](0) + [tGEF](0) [tGEF^*](0)$. Hence, the total amount of mG protein must be initially higher than the concentration of tGEF and the summed concentrations of total tG and total tGEF proteins must be initially higher than the concentration of mGAP. Finally, the inequality $0 \leq \xi + [mGAP_{tot}] \leq \tilde{C}_2$ indicates that initial balance between the summed concentrations of total tG, total tGEF and the available mGAP is the upper bound for the tG concentrations.  For Family 3 (Eq. \ref{ss3_123}), where mG and mG* have zero concentration values and the tGTPase is fully inactivated, the inequality $\tilde{C_1} \leq \tilde{C_2}$ becomes $[mG](0) +[mG^*](0) \leq [tGEF](0)$. As for Family 1, the total amount of mG protein must be initially less than the concentration of tGEF. Moreover, from $\tilde{C}_1 - [mGAP_{tot}] \leq \xi$,  the initial balance between the summed concentrations of total mG, total tG, tGEF* and the available mGAP is the lower bound for the tG concentration. For Family 4 (Eq. \ref{ss4_123_new}), where mG and mG* are present and tG* concentration is zero, $\tilde{C_2} \leq  \tilde{C_1}$ becomes $[mG](0) +[mG^*](0) \geq [tGEF](0)$. As for Family, 2 the total amount of mG protein must be initially higher than the concentration of tGEF. Moreover, from $\tilde{C}_2 - [mGAP_{tot}] \leq \xi$, the  initial balance between the summed concentrations of total tG, total tGEF and the available mGAP is the lower bound for the tG concentrations.

Fig.\ref{Fig2} illustrates the four Families (gray-colored ``1+2+3'' in the $2 \times 2$ table) promoted by the feedforward connection $mG^* \to tGEF$ and the feedback loops $tGEF \to mGAP$ and $tG^* \to mGAP$. Families 1 and 2 have a similar interpretation of the steady states 1 and 2 obtained in \Cref{subsec_Arrow1} and \Cref{subsec_Arrows12}. On the other hand, Families 3 and 4 were obtained through contributions of the feedback loop $tG^* \to mGAP$. These states share the common feature of having tGTPase fully inactivated. As for steady states 3 and 4,   Family 3 can be interpreted as a configuration where the copy numbers of mG and mG* are low, while in Family 4, those copy numbers are high.

\begin{figure}
	\centering
%	\hspace{-1cm}
	\includegraphics[scale=0.25]{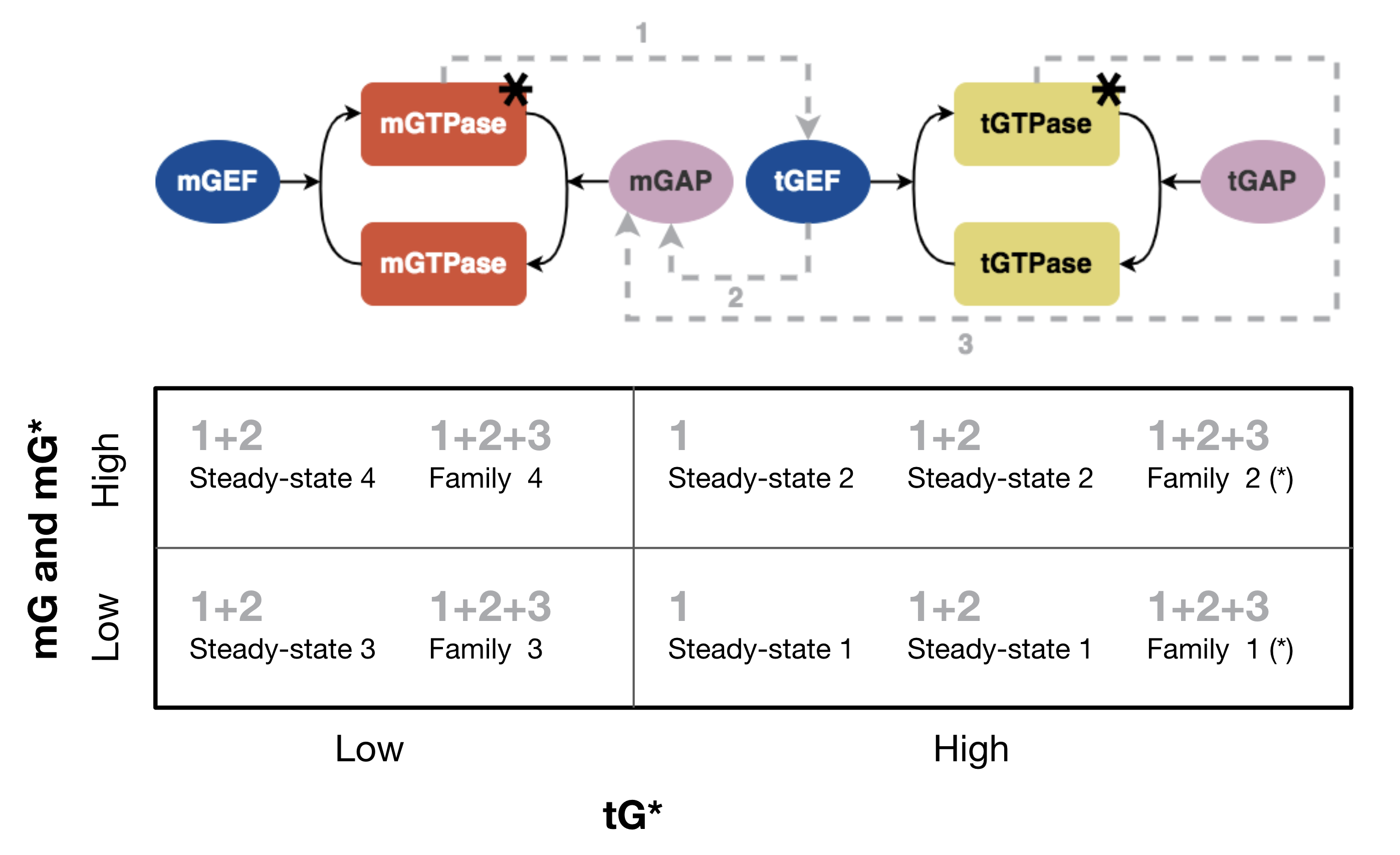}
	\caption{\textbf{Possible steady states promoted by the feedforward connection and the feedback loops}. For the three combinations of arrows (1, 1+2, and 1+2+3) chosen in our study, we calculate the steady state solutions for the coupled GTPase circuit model. We then characterize four steady state configurations that emerge from our analysis, depending on low/high copy numbers of the m- and tGTPases. The feedforward connection $mG^* \to tGEF$ (represented by ``1'') allows for the emergence of two steady states with low/high mG and mG* concentrations, while tG* steady state concentration remained high in both cases. On the other hand, low tG* steady state concentrations were obtained when the feedback loop $tGEF \to mGAP$ was added to the system (represented by ``1+2''). Finally, the feedback loop  $tG^* \to mGAP$ allowed for the emergence of four parametrized families of steady state within the same low/high configurations. (*) Families 2 and 4 can have higher tG* steady states provided the inactive tGTPase concentrations assume strictly positive values.
	 }
 \label{Fig2}
\end{figure}

\subsection{Numerical Simulations.}
\label{num_sim}

To complete our mathematical analysis, we numerically investigate the range of initial conditions in which the trajectories of the system converge to the different steady states. In particular, we illustrate the so-called \emph{basins of attraction} \cite{nusse2012dynamics} of the steady states, considering the same combination of connections between the two GTPase switches from \S \ref{subsec_Arrow1} -- \S \ref{subsec_Arrows123}. 
In Fig.\ref{Fig_BA1}, we explore the case where the two GTPase switches are coupled through the  feedforward connection $mG^* \to tGEF$ ( Fig.\ref{Fig_BA1}A). 
We color the trajectories of the system according to the comparison between the initial conditions $[mG](0) + [mG^*](0)$ and $[tGEF^*](0)$ from the steady state analysis in \Cref{subsec_Arrow1}.
For fixed $[mG^*](0)$ and  $[tGEF](0)$ values, we consider $[mG](0)$ ranging from 0 to 10 $\mu M$ and therefore $[mG](0) + [mG^*](0)$ can be less of higher than $[tGEF](0)$ (blue or red-colored lines and dots). For all simulations, we plot the trajectories of the system until equilibrium is reached. 
In the $\mathcal{G}^* \times [mG_{tot}]$ plane (Fig.\ref{Fig_BA1}B), we observe a linear relationship between these two quantities, where the black arrows indicate the time direction.
If  $[mG](0) + [mG^*](0) < [tGEF](0)$, the system converges to a state where no active mGTPase exists (blue colored trajectories in Figs. \ref{Fig_BA1}B and \ref{Fig_BA1}C). 
On the other hand, if  $[mG](0) + [mG^*](0) > [tGEF](0)$, the system converges a state where  the concentration of the active and inactive mGTPase are positive at the final time (red-colored trajectories). 
To visualize these results in terms of dose-response curves, in Fig.\ref{Fig_BA1}D we plot the final-state values of $[mG_{tot}]$ and $\mathcal{G}^*$ (denoted by s.s) as a function of $[mG](0)$. 
The trajectories in the $\mathcal{T}^* \times [mG_{tot}]$ plane are shown in Fig.\ref{Fig_BA1}E. We observe a detail showing that $\mathcal{T}^*$ reaches a fixed final value around 0.97 when $[mG](0) + [mG^*](0) > [tGEF](0)$ (see magnified view).
We observe that the trajectories converge to steady states that agree with the local stability results from \S \Cref{subsec_Arrow1}.  This suggests that the conditions $[mG](0) + [mG^*](0)<[tGEF](0)$  and $[mG](0) + [mG^*](0)>[tGEF](0)$ are not only valid in a neighborhood of the steady states, but also hold for other initial values satisfying those inequalities.

\begin{figure}
	\centering
	%\hspace{-5cm}
	\includegraphics[scale=0.75]{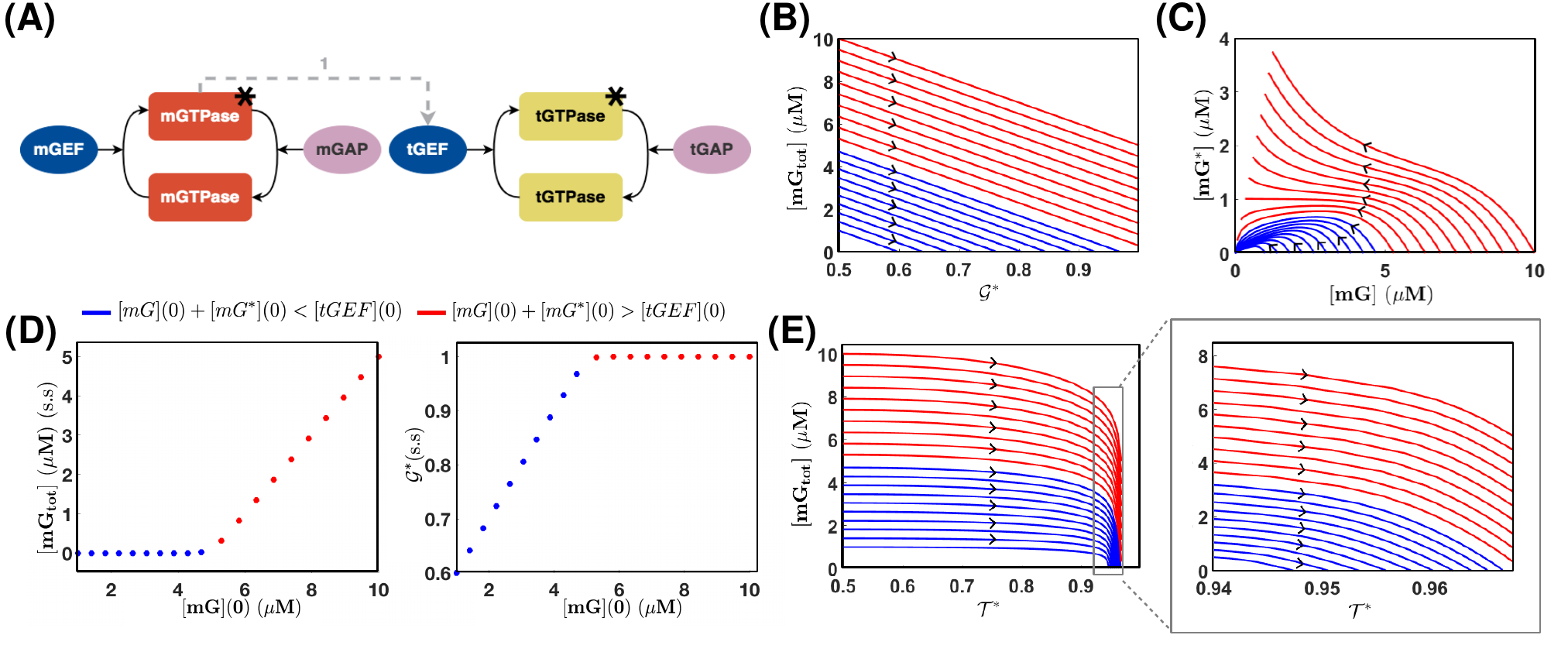}
	\caption{\small{\textbf{Trajectories of the system and steady states (arrow 1} (A) Schematics with the coupled GTPase switches and a feedforward connection $mG^* \to tGEF$, represented by arrow 1. (B)  $[mG](0)$ was changed from 0 to 10 $\mu M$ and the trajectories of the system were calculated until equilibrium was reached. In the $\mathcal{G}^* \times [mG_{tot}]$ plane, a linear relationship emerges. The black arrows indicate the direction of time. If $[mG](0) > 5 \mu M$, the system converges to a final-state where the concentrations of the active and inactive mGTPase are nonzero. On the other hand, when $[mG](0) < 5 \mu M$, the trajectories converge a final-state with no mGTPase exists (blue colored lines). (C) Trajectories of the active ($[mG^∗]$) vs inactive mGTPase ([mG]) for $[mG](0)$. (D) Dose response curves show the steady states (denoted by s.s) for the total mGTPase concentration and fraction of active tGEF ($\mathcal{G}^∗$ ) depending on $[mG](0)$ in the two different scenarios. (E) The dynamics in the $\mathcal{T}^∗ \times [mG_{tot}]$ plane. Parameter values: $[mG^*](0)= 0 \mu M, \mathcal{T}^*(0) = 0.5, \mathcal{G}^*(0) = 0.5, [mGAP^∗] = 1 \mu M, [mGEF^∗] = 1 \mu M, k_{on} = 3 (s.\mu M)^{-1}, k_{off} = 1 (s.\mu M)^{-1}, [tGAP^∗] = 1 \mu M, [tGEF_{tot}] = 10 \mu M, [tG_{tot}] = 10 \mu M$. Simulation time: $5 s$ for panels B and E, $50 s$ for panels C, D, F, and G. Numerical simulations were performed using the solver ode15s in Matlab R2018a. All parameters were arbitrarily chosen only to illustrate the dynamic features of the model.}}
	\label{Fig_BA1}
\end{figure}

Fig.\ref{Fig_panel12} illustrates the dynamics of the system when the feedback loop $tGEF \to mGAP$ (Arrow 2) is added to the feedforward connection (Fig.\ref{Fig_panel12}A). In Fig.\ref{Fig_panel12}B, we plot several $[tGEF^*]$ trajectories starting at $[tGEF^*] = 5 \mu M$ for different $[mG](0)$ and $[mGAP](0)$ values. The resulting rich variety of curves indicate the sensitivity of the system to these initial conditions.  In Fig.\ref{Fig_panel12}C, different dose-response curves are generated to show the steady state tGEF* values. If $[mGAP](0)=0 \mu M$ (blue and red dots), only the feedforward connection affects the system, since mGAP cannot be activated by tGEF*.  When $[mGAP](0)=1$ (green squares),  a similar steady state profile emerges, with $[tGEF^*]$ s.s increasing for $[mG](0) \leq 5 \mu M$ and remaining constant $[mG](0)> 5$. When $[mGAP](0)=8 \mu M$, $[tGEF^*]$ is zero for $[mG](0)<2 \mu M$ and increases until $[mG](0)<5 \mu M $. For $[mG](0)>5$, the steady state achieves its maximum value  slightly above $[tGEF^*]>2$. Finally when $[mGAP](0)=11 \mu M$, $tGEF^*$ becomes fully recruited by mGAP and the $[tGEF^*]$ s.s is zero for all $[mG](0)$ values. In Fig.\ref{Fig_panel12}D, we  scan the space of initial amounts of mG and mGAP. When $[mGAP](0)> 10 \mu M$, the $[tGEF^*]$ s.s is zero, while for $[mGAP](0)< 10 \mu M$ is becomes nonzero and dependent of $[mG](0)$. In Figs. \ref{Fig_panel12} E, F, and G, we analyze the $[tG^*]$ concentration values and obtain similar results.  

\begin{figure}
	\centering
	%\hspace{-5cm}
	\includegraphics[scale=1]{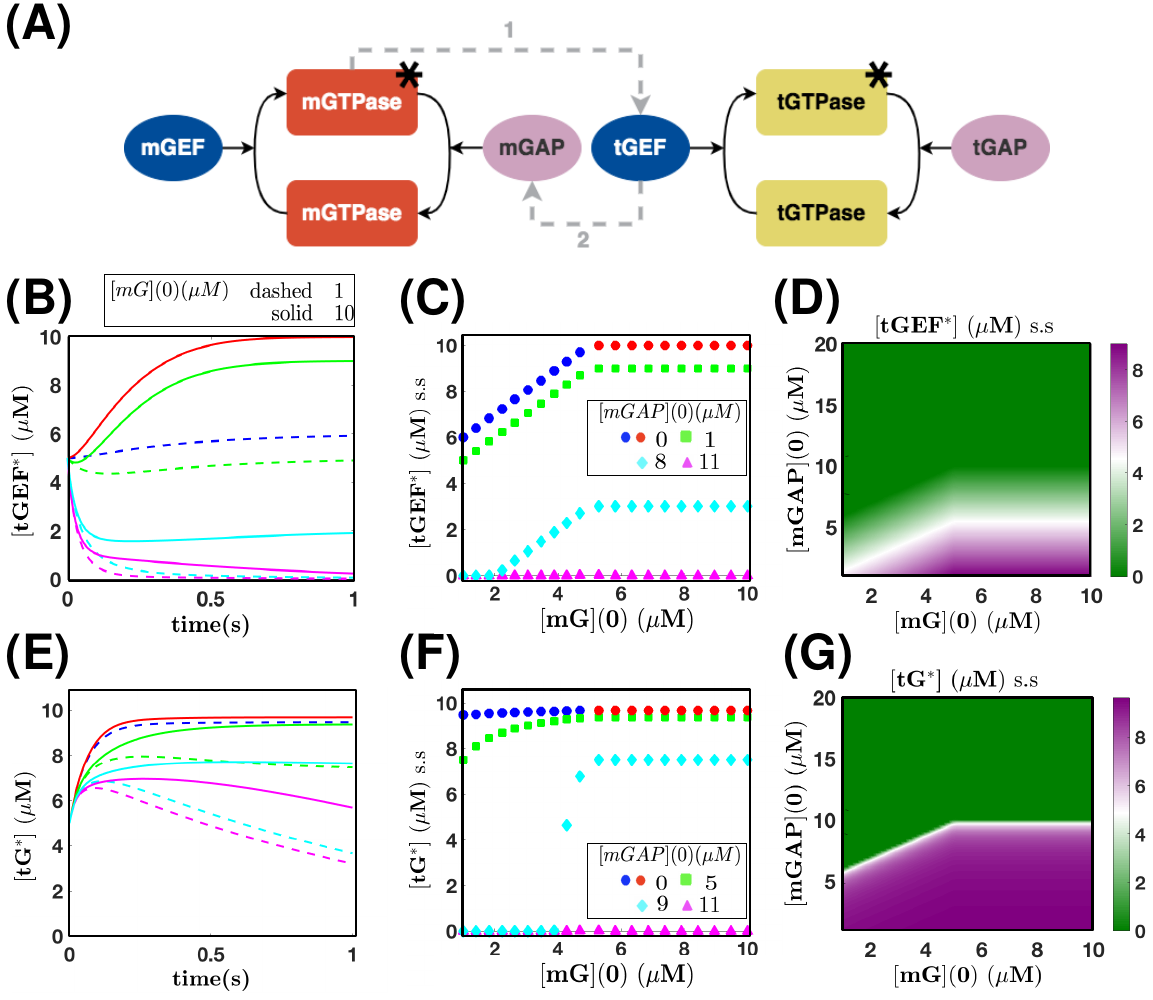}
	\caption{\footnotesize{\textbf{Trajectories of the system and steady states (s.s) (arrows 1 and 2).} (A) Schematics of the coupled GTPases with the feedforward connection $mG^* \to tGEF$ (arrow 1) and the feedback loop $tGEF \to mGAP$ (arrow 2). (B) $[tGEF^*]$ trajectories for $[mGAP](0) = $ 0, 1, 8, and 11 $\mu M$. For each $[mGAP](0)$ value, we plot two curves for $[mG](0) = $ 1 (dashed) and 10 $\mu M$ (solid) (C) Dose response curves show $[tGEF^*]$ s.s when $[mG](0)$  ranges from 0 to 10 $\mu M$. If $[mGAP](0)=0 \mu M$  (blue and red dots),  there will be no mGAP activation and therefore no effects of the feedback. For $[mGAP](0)>0 \mu M$, the feedback becomes effective and generate different $[tGEF^*]$   responses.  (D) Colormap for $[tGEF^*]$  s.s concentrations for a range of $[mG](0)$ and $[mGAP](0)$ values.   A sharp decrease on $[tGEF^*]$  occurs  when $[mGAP](0) \geq 10 \mu M$. When $[mGAP](0) < 10 \mu M$, the $[tGEF^*]$ s.s depend on $[mG](0)$. (E) $[tG^*]$ trajectories for $[mGAP](0) = $ 0, 5, 9, and 11 $\mu M$ and same $[mG](0)$. (F) Dose response curves for $[tG^*]$ s.s depend on $[mGAP](0)$. (G) Colormap for  $[tG^*]$ s.s.; lower tG* concentrations result from higher $[mGAP](0)$ values, since tGEF* is recruited for mGAP activation. Parameter values: $k_{on} =3 (s.\mu M s)^{-1}$, $k_off  = 1 (s.\mu M)^{-1}$, $[mG^*](0) = 0 \mu M$, $[tGEF_{tot}](0) = 10 \mu M$, $[tGEF*](0) = 5 \mu M$, $\mathcal{T}^*(0) = 0.5$, $[tG_{tot}] = 10 \mu M$, $[mGAP^*](0)=1 \mu M$, $[tGAP^*](0) = 1 \mu M$, $[mGEF^*]=1 \mu M$.  Simulation times: $5 s$ (B and E) and $50 s$ (C, D, F, and G). Numerical simulations were performed using the solver ode23s in Matlab R2018a. All parameters were arbitrarily chosen only to illustrate the dynamic features of the model.}}
	\label{Fig_panel12}
\end{figure}

Fig.\ref{Fig_panel123} illustrates the dynamics of the system when the feedback loops $tGEF \to mGAP$ and $tG^* \to mGAP$ are added to the feedforward connection (Fig.\ref{Fig_panel123}A). In Fig.\ref{Fig_panel12}B, we plot several $[tGEF^*]$ trajectories starting at $[tGEF^*] = 5 \mu M$ for different $[mG](0)$ and $[mGAP](0)$ values.  In Fig.\ref{Fig_panel12}C, different dose-response curves are generated to show the steady state tGEF* values. As in the previous case with only one feedback loop, if $[mGAP](0)=0 \mu M$ (blue and red dots), mGAP cannot be activated by tGEF*.  When $[mGAP](0)=1$ (green squares),  a similar steady state profile emerges, with $[tGEF^*]$ s.s increasing for $[mG](0) \leq 5 \mu M$ and remaining constant $[mG](0)> 5$. When $[mGAP](0)=8$ and  11 $\mu M$, $[tGEF^*]$ increases until $[mG](0)<5 \mu M $. For $[mG](0)>5$, the steady state achieves its maximum value. In Fig.\ref{Fig_panel123}D, we scan the space of initial amounts of mG and mGAP and we observe a more graded response in comparison with Fig.\ref{Fig_panel12}. In Figs. \ref{Fig_panel123} E, F, and G, we analyze the $[tG^*]$ concentration values and obtain similar results.  

\begin{figure}
	\centering
	%\hspace{-5cm}
	\includegraphics[scale=1]{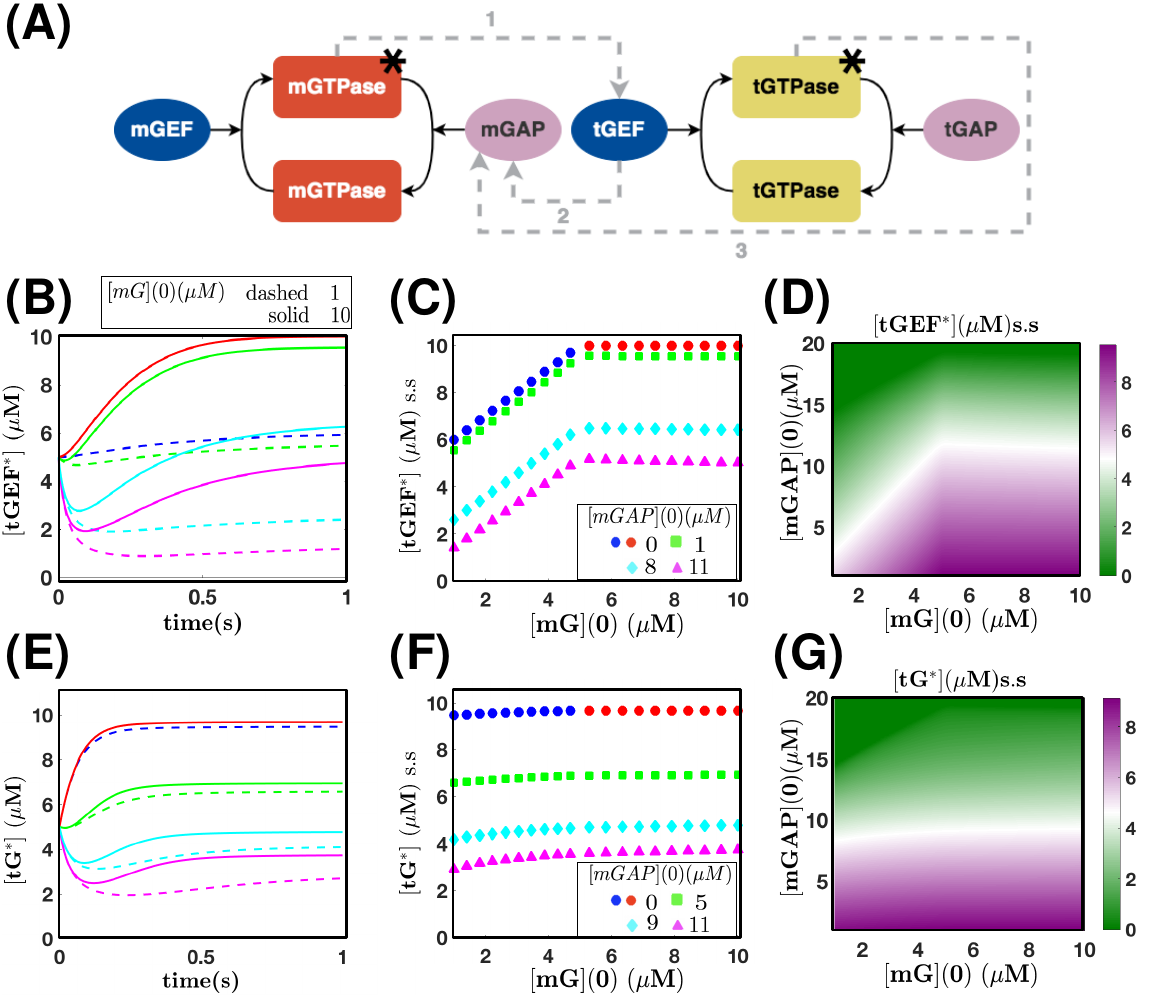}
	\caption{\footnotesize{\textbf{Trajectories of the system and steady states (s.s) (arrows 1, 2, and 3).} (A) Schematics of the coupled GTPases with the feedforward connection $mG^* \to tGEF$ (arrow 1) and the feedback loops $tGEF \to mGAP$ (arrow 2) and $tG^* \to mGAP$ (arrow 3). (B) $[tGEF^*]$ trajectories for $[mGAP](0) = $ 0, 1, 8, and 11 $\mu M$. For each $[mGAP](0)$ value, we plot two curves for $[mG](0) = $ 1 and 10 $\mu M$. (C) Dose response curves show $[tGEF^*]$ s.s when $[mG](0)$  ranges from 0 to 10 $\mu M$. If $[mGAP](0)=0 \mu M$  (blue and red dots),  there will be no mGAP activation and therefore no effects of the feedback loops. For $[mGAP](0)>0 \mu M$, the feedback becomes effective and generate different $[tGEF^*]$   responses.  (D) Colormap for $[tGEF^*]$  s.s concentrations for a range of $[mG](0)$ and $[mGAP](0)$ values. A more graded decrease on $[tGEF^*]$  occurs  when $[mGAP](0) \geq 10 \mu M$ in comparison with Fig.\ref{Fig_panel12}D. (E) $[tG^*]$ trajectories for $[mGAP](0) = $ 0, 5, 9, and 11 $\mu M$ and same $[mG](0)$. (F) Dose response curves for $[tG^*]$ s.s  depend on $[mGAP](0)$ and does not change significantly as $[mG](0)$ increases. (G) Colormap for  $[tG^*]$ s.s.; lower tG* concentrations result from higher $[mGAP](0)$ values, since tGEF* and tG* are recruited for mGAP activation. Parameter values: $k_{on} =3 (s.\mu M s)^{-1}$, $k_off  = 1 (s.\mu M)^{-1}$, $[mG^*](0) = 0 \mu M$, $[tGEF_{tot}](0) = 10 \mu M$, $[tGEF*](0) = 5 \mu M$, $\mathcal{T}^*(0) = 0.5$, $[tG_{tot}] = 10 \mu M$, $[mGAP^*](0)=1 \mu M$, $[tGAP^*](0) = 1 \mu M$, $[mGEF^*]=1 \mu M$. Simulation times: $5 s$ (B and E) and $50 s$ (C, D, F, and G) . Numerical simulations were performed using the solver ode23s in Matlab R2018a. All parameters were arbitrarily chosen only to illustrate the dynamic features of the model.}}
	\label{Fig_panel123}
\end{figure}

\color{black}
In Fig.\ref{Fig_BA12_123}, we investigate the space of initial conditions for mG* and mGAP* in which the system converges to the different steady states. Fig.\ref{Fig_BA12_123}A shows the simplest system where the two GTPase switches are connected by the feedforward $mG^* \to tGEF$. Two steady states are obtained depending on the initial amount of mG*. For $[mG^*](0)< [tGEF](0)- [mGAP^*](0) = 5 \mu M$, the trajectories converge to steady state 1 with no mG and mG* concentrations. On the other hand, for $[mG^*](0)> [tGEF](0)- [mGAP^*](0) = 5 \mu M$, then the system achieves the steady state 2 with non zero concentrations of both m and t-GTPase. Fig.\ref{Fig_BA12_123}B shows the results for the feedforward connection, and feedback loops $tGEF \to mGAP$ (arrows 1+2). In this particular example, the four steady states can be achieved for $[mGAP^*](0)$ and $[mG^*](0)$  ranging from 0 to 12 $\mu M$ and 0 and 10 $\mu M$, respectively. In the vertical direction, the initial amount of mG* governs the transitions from steady states 3 to 4 (lower $[mGAP^*](0)$) and 1 to 2 (higher $[mGAP^*](0)$). In both steady states 2 and 4 (Eqs. \ref{ss2_12} and \ref{ss4_12_new}), the concentrations of mGTPase are nonzero. Therefore, we predict that an increase of initial concentration of mG* would favor the emergence of these two steady states. In the horizontal direction, when the initial amount of mGAP* increases,  the available mGAP (inactive) decreases as we set the total mGAP as 12 $\mu M$, which reduces the effects of the feedback loops and thus facilitates the emergence of the steady states 1 and 2 where the concentrations of tGTPase are nonzero. 

Fig.\ref{Fig_BA12_123}C shows a similar colormap for the system with both feedback loops $tGEF \to mGAP$ and $tG^* \to mGAP$. It is worth noticing the expansion of the basin of attraction of Families 1 and 2 compared to Fig.\ref{Fig_BA12_123}A, while the basin of Families 3 and 4 shrinks. Remarkably, in both Figs. \ref{Fig_BA12_123}B and \ref{Fig_BA12_123}C, there is a critical point  (represented by a black cross) of intersection of the four basins of attraction. In this case, disturbances in the initial conditions around that intersection point can drive the system to different steady states. Thus, while coupling the two GTPase switches with a forward arrow only gives two possible steady states, the negative feedback afforded by arrows 2 and 3 give rise to a larger range of possibilities. Additionally, the existence of a critical point emerges in the presence of the negative feedback suggesting a rich phase space for this coupled system.

\begin{figure}
	\centering
	%\hspace{-3.35cm}
	\includegraphics[scale=0.55]{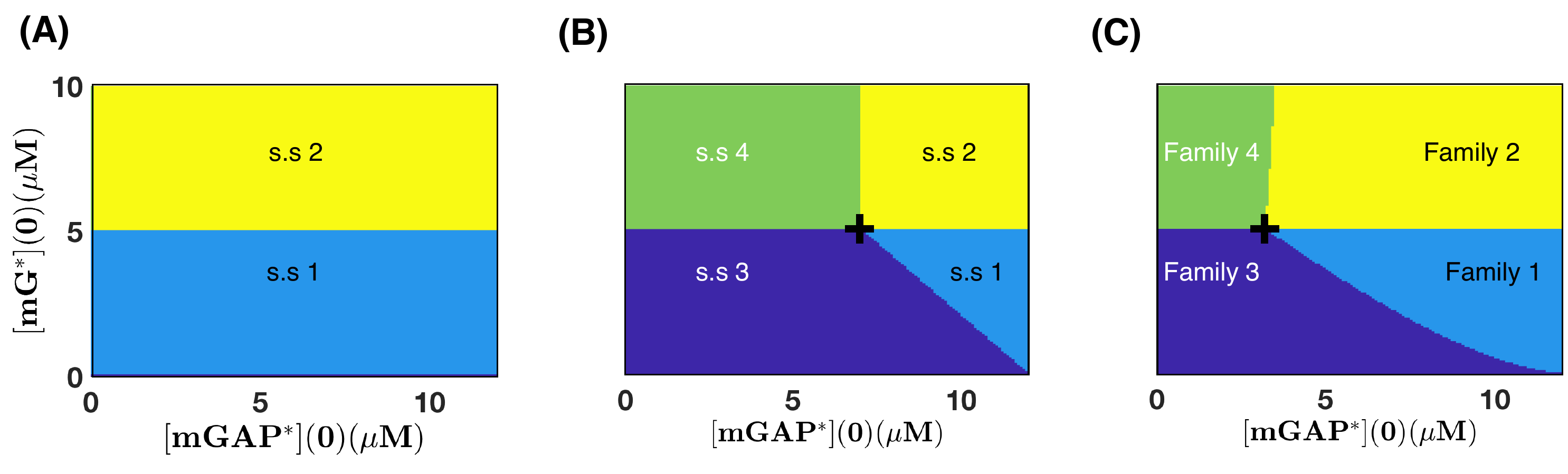}
	\caption{\small{\textbf{Basins of Attraction  -- dependency on $[mG^*(0)]$ and $[mGAP^*(0)]$} (A) The two steady states of the system with feedforward connection (\Cref{subsec_Arrow1})  are only driven by changes in the initial amount of mG* (B) When the feedforward and both feedback loops $tGEF \to mGAP$  are considered, we observe the emergence of four regions (green,yellow, dark blue and light blue colored) corresponding to the four  steady states from \Cref{subsec_Arrows12} (C) A similar result was found when we analyzed the system with the feedforward and both feedback loops $tGEF \to mGAP$ and $tG^* \to mGAP$. A black cross  indicates a critical point at the intersection of the four basins of attraction. Parameter values: $k_{on}  = 3 (s.\mu M)^{-1}$, $k_{off} = 1 (s.\mu M)^{-1} $, $[mG](0) = 0 \mu M$, $[mG^*](0) = 0 \mu M$, $[tG](0) = 5 \mu M $, $[tG^*](0) = 0 \mu M$, $[tGEF](0)=5 \mu M$, $[mGAP](0)=12 \mu M - [mGAP^*](0) $, $[tGAP^*](0) = 1 \mu M$, $[mGEF^*](0)  = 1 \mu M$}}
	\label{Fig_BA12_123}
\end{figure}

\section{Discussion}
\label{discussion}

\begin{figure}
	\centering
	%\hspace{-5cm}
	\includegraphics[scale=0.25]{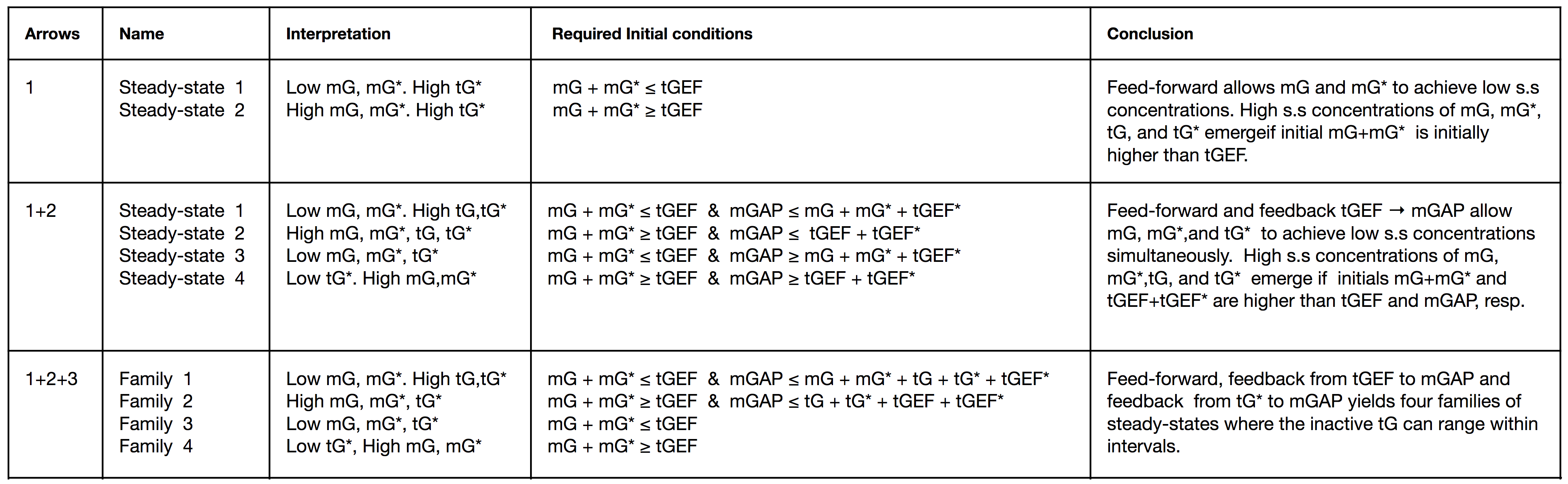}
	\caption{\textbf{Main results and conclusions from steady state analysis} We performed a steady state analysis of a GTPase coupled circuit that has been observed experimentally. For three biologically relevant combinations among the feedforward and two feedback loops, we present the steady states and their interpretation. Moreover, we established the required initial conditions for the existence of the steady states. Each connection  adds to the richness of the functioning of these coupled GTPase switches.}
	\label{Fig3}
\end{figure}

GTP-binding proteins (GTPases) regulate crucial aspects of numerous cellular events. Their ability to act as biochemical switches is essential to promote information processing within signaling networks. The two types of GTPases  - monomeric (m) and heterotrimeric (t) - have traditionally believed to function independently. More recently, a series of experimental findings revealed that m- and tGTPases co-regulate each other in the Golgi through a functionally coupled circuit \cite{lo2015activation}. In this work, we sought to understand the dynamic properties of this dual GTPase circuit. To this end, we developed a system of differential equations that describes the evolution of two coupled GTPase switches. Our analysis provided insight into the emergence of steady state configurations that cannot be achieved in systems of isolated GTPase switches. To the best of our knowledge, this is the first modeling effort that described coupled GTPase switches with a feedforward connection $mG^* \to tGEF$ and the two feedback loops $tGEF \to mGAP$ and $tG^* \to mGAP$ \cite{lo2015activation,jamora1997regulation,stow1991heterotrimeric,stow1998vesicle,stow1995regulation,cancino2013signaling}.   
 
A major result from our analysis is a systematic characterization of the steady state concentrations of both m- and tGTPases, as well as their GEFs and GAPs.  Given the model formulation and the fact that we do not know the various kinetic parameters, obtaining these states is critical. We show the obtained steady states in all three arrow combinations that we chose to analyze (Fig.\ref{Fig3}). Remarkably, the different steady states show a variety of configurations in which both m- and tGTPase can be interpreted as having low or high concentration values. These results help us to understand how the feedforward connection and feedback loops impact the coupled system. Once we know the steady states, locally asymptotic stability ensures that the trajectories will converge back to the steady state, given sufficiently close initial conditions. In  \S \ref{subsec_Arrow1} -- \S \ref{subsec_Arrows12}, we confirmed that all steady states obtained with a feedforward connection and feedback loop $tGEF \to mGAP$ (arrow 1 or arrows 1+2 in Fig.\ref{Fig1}B) are locally asymptotically stable. However, when the two feedback loops are considered along with the feedforward (arrows 1+2+3 in Fig.\ref{Fig1}B), the local stability analysis cannot be performed because the steady states are not isolated. Instead, we obtain four one-parameter families that depend on the amount of inactive tGTPase. At this point, further investigation would be needed to determine the behavior of the system near those steady state families. Even as we aim to develop complex models that are refined with  iterative experimental validations, we note that our analysis gives insight to different steady states that emerge due to different couplings that may not exist in physiology. Such insights may become meaningful in the context of disease pathogenesis where copy numbers of each player in the network motif may change relative to each other, and do so dynamically (e.g., when responding to stress/stimuli), or disease-driving mutations alter their functions (e.g., activating and inactivating mutations in GTPases, GAPs, or GEFs).    

Even though the mathematical structure of the model may appear simple, we find a rich phase space for the coupled GTPase switches by analyzing the combination of network connections that have more biological meaning. Future studies could also explore the role of an external stimulus in the emergence of ultrasensitive behavior \cite{lipshtat2010design} and spatial organization of these switches. 
Moreover, our system is well suited for coupling with experimental measurements \cite{getz2019predictive}, including dose-response curves, response times, and noise fluctuations, as done recently in \cite{ghusinga2020molecular}.

\newpage

\section{Acknowledgments} 

This work was supported by Air Force Office of Scientific Research (AFOSR)
Multidisciplinary University Research Initiative (MURI) grant
FA9550-18-1-0051 (to P. Rangamani) and the National Institute of Health (CA100768, CA238042 and AI141630 to P. Ghosh). Lucas M. Stolerman acknowledges support from the National Institute of Health (CA209891).

% Bibliography
\footnotesize{
\bibliography{IB_ref_motif} %your .bib file
\bibliographystyle{unsrtPR_motif}
}
\clearpage

%%%%%%%%%%%%% %%%%%%%%%%%%%% APPENDIX %%%%%%%%%%%%%%%%%%%%%%%%%%%%%%%

%\begin{comment}
\newpage
\appendix

%%%%%%%%%%%%%%%%%%%%%%%%% PROOF OF PROPOSITION ARROW 1 %%%%%%%%%%%%%%%%%%%%%%%%%%%%%%%%%%
\normalsize
\section{Proof of Proposition \ref{prop_arrow1}}
\label{appendix_1}

We must find nonnegative $\widehat{[mG]}$, $\widehat{[mG^*]}$, $\widehat{\mathcal{G}^*}$, and $\widehat{\mathcal{T}^*}$ satisfying the following system:  
\small
\begin{eqnarray}
%	\label{model_red_2}
	%\left \{ 
	%\begin{array}{l}
 &&  k_{on} [mGEF^*]\widehat{[mG]}  - k_{off}[mGAP^*]\widehat{[mG^*]} -k_{on}[tGEF_{tot}](1 - \widehat{\mathcal{G}^*})\widehat{[mG^*]} = 0 \label{steadystate1:1}  \\
&& k_{on} [tGEF_{tot}]  \widehat{\mathcal{G}^*} (1 - \widehat{\mathcal{T}^*}) - k_{off}[tGAP^*]\widehat{\mathcal{T}^*}   = 0   \label{steadystate1:2} \\
 && k_{on}(1 -\widehat{\mathcal{G}^*})\widehat{[mG^*]} = 0 \label{steadystate1:3} \\
 && 
 \widehat{[mG]} +  \widehat{[mG^*]} + [tGEF_{tot}] \widehat{\mathcal{G^*}} = C.  \label{steadystate1:4}
%	\end{array} \right.
\end{eqnarray}
\normalsize

 From Eq. \ref{steadystate1:3}, and since we assume $k_{on}>0$, we must have   $\widehat{[mG^*]}=0$ or $\widehat{\mathcal{G}^*} = 1 $. Thus we divide the steady state analysis in two cases.

\noindent
\underline{Case 1:  $\widehat{[mG^*]}=0.$}
\noindent
\\
From Eq. \ref{steadystate1:1} we must have  $\widehat{[mG]}=0$ and from Eq. \ref{steadystate1:4}, we obtain
$ \widehat{\mathcal{G^*}} = \frac{C}{[tGEF_{tot}]} $. Since $\widehat{\mathcal{G^*}} \leq 1 $ by definition, we conclude that 
\begin{equation}
C \leq [tGEF_{tot}].
\label{cond1}
\end{equation}
Eq. \ref{cond1} is also sufficient for  $\widehat{[mG^*]}=0$. Otherwise, if $C \leq [tGEF_{tot}]$ and $\widehat{[mG^*]}>0$, then $\widehat{\mathcal{G^*}} = 1$ (Eq. \ref{steadystate1:3}) and from Eq. \ref{steadystate1:4}, we would conclude that $\widehat{[mG]}+\widehat{[mG^*]} \leq 0$, which is imposible. 

Finally, by substituting $ \widehat{\mathcal{G^*}}$ in  Eq. \ref{steadystate1:2}, %\todopr{check your labels?}, 
we obtain $ \widehat{\mathcal{T}^*} = \frac{1}{1+\frac{k_{off} [tGAP^*]}{k_{on} C}}$
and therefore the steady state is given by
\begin{equation*}
\left(\widehat{[mG]},\widehat{[mG^*]},\widehat{\mathcal{T}^*},\widehat{\mathcal{G}^*}  \right) = \left(0, 0, \frac{1}{1+\frac{k_{off} [tGAP^*]}{k_{on} C}}, \frac{C}{[tGEF_{tot}]}  \right)
%\label{ss1}
\end{equation*}

\noindent
\underline{Case 2:  $\widehat{\mathcal{G^*}}=1$}
\noindent
\\
In this case, $ \widehat{[mG^*]} \geq 0$ and from  Eqs. \ref{steadystate1:1} and \ref{steadystate1:4},%\todopr{labels},
we obtain 
\begin{equation*}
\widehat{[mG^*]} = \frac{k_{on} [mGEF^*]}{k_{on}[mGEF^*]+k_{off}[mGAP^*]} \left( C- [tGEF_{tot}]\right) 
\end{equation*}
and
\begin{equation*}
\widehat{[mG]} = \frac{k_{off} [mGAP^*]}{k_{on}[mGEF^*]+k_{off}[mGAP^*]} \left( C- [tGEF_{tot}]\right).
\end{equation*}
In this case, since the steady state has to be nonnegative, we must have 
\begin{equation}
C \geq [tGEF_{tot}].
\label{cond2}
\end{equation}
which is also sufficient for $\widehat{\mathcal{G^*}}=1$. Otherwise if $C \geq [tGEF_{tot}]$ and $\widehat{\mathcal{G^*}}<1$, then $\widehat{[mG^*]}=\widehat{[mG]}=0$ (Eq. \ref{steadystate1:1}) and, from Eq. \ref{steadystate1:3}, we would have 
$$  C =  \widehat{[mG]} +  \widehat{[mG^*]} + [tGEF_{tot}] \widehat{\mathcal{G^*}} <  [tGEF_{tot}], $$ which is impossible.

Finally, by substituting $\widehat{\mathcal{G^*}}=1$ in Eq. \ref{steadystate1:2},  we obtain
$$  k_{on} [tGEF_{tot}] (1 - \widehat{\mathcal{T}^*}) - k_{off}[tGAP^*]\widehat{\mathcal{T}^*}= 0$$
which gives  us  $\widehat{\mathcal{T}^*} = \frac{k_{on} [tGEF_{tot}]}{k_{on}[tGEF_{tot}] + k_{off}[tGAP^*]}$
and therefore

\small
\begin{eqnarray}
\left(\widehat{[mG]},\widehat{[mG^*]},\widehat{\mathcal{T}^*},\widehat{\mathcal{G}^*}  \right) &=& 
\left(\frac{k_{off} [mGAP^*]}{k_{on}[mGEF^*]+k_{off}[mGAP^*]} \left( C- [tGEF_{tot}]\right), \right. \nonumber\\
&& \left. \frac{k_{on} [mGEF^*]}{k_{on}[mGEF^*]+k_{off}[mGAP^*]} \left( C- [tGEF_{tot}]\right), \right. \nonumber \\
&& \left. \frac{k_{on} [tGEF_{tot}]}{k_{on}[tGEF_{tot}] + k_{off}[tGAP^*]},1  \right).
%\label{ss2}
\end{eqnarray}
\normalsize

%%%%%%%%%%%%%%%%%%%%%%%%%%%%%%% PROOF OF THEOREM ARROWS 12 %%%%%%%%%%%%%%%%%%%%%%%%%%%%%
\newpage
\section{Proof of Theorem \ref{prop_12}}
\label{appendix_12}
\normalsize

We start our proof by computing the steady states of the system, which are solutions of the algebraic system given by Eqs. \ref{ss_red12:1}--\ref{ss_red12:7}. 
We also establish necessary and sufficient conditions involving the parameters $C_1$, $C_2$, and $[mGAP_{tot}]$ for the existence of each steady state.
We then compute the Jacobian matrix of the system and determine the local stability of the steady state based on the classical linearization procedure \cite{strogatz1994}.

\subsection*{Steady states.}  
 We divide our analysis into four different cases that emerge from the preliminary inspection of the system given by  Eqs. \ref{ss_red12:1}--\ref{ss_red12:7}.

\normalsize
\underline{Case 1:  $\widehat{[mG^*]}=0$ and $\widehat{[mGAP^*]}=[mGAP_{tot}]$.}
\noindent
\\
From Eq. \ref{ss_red12:1}, we have $\widehat{[mG]} = 0$ and from Eq. \ref{ss_red12:6}, $\widehat{[tGEF^*]} = C_1 - [mGAP_{tot}]$.
Thus $C_1 \geq [mGAP_{tot}]$ since the steady state must be nonnegative. Now Eq. \ref{ss_red12:7} gives $\widehat{[tGEF]} = C_2 - C_1$ and that implies $C_2 \geq C_1$.  

Finally, Eq. \ref{ss_red12:3} yields
$$  (C_1 - [mGAP_{tot}])(1- \widehat{\mathcal{T}^*})  -  k [tGAP^*] \widehat{\mathcal{T}^*} = 0 $$
and hence 
$$ \widehat{\mathcal{T}^*} = \frac{C_1 - [mGAP_{tot}]}{(C_1 - [ mGAP_{tot}]) + k [tGAP^*]} $$
The steady state is therefore given by 
\footnotesize
\begin{eqnarray}
\widehat{\bf{x}}
&=& \bigg(0,0,\frac{C_1 - [mGAP_{tot}]}{(C_1 - [mGAP_{tot}]) + k [tGAP^*]}, C_2 - C_1,C_1 - [mGAP_{tot}],[mGAP_{tot}] \bigg).
\nonumber
\end{eqnarray}

%%%%%%%%%%%%%%%%%%% Sufficient conditions  old SS case 4 --- > new case 1
\normalsize
We now observe that the two parameter relations
\begin{equation}
C_1 \geq [mGAP_{tot}] \quad \text{and} \quad C_2 \geq C_1
\label{cond4_12}
\end{equation}
 are sufficient for $\widehat{[mG^*]}=0$ and $\widehat{[mGAP^*]}=[mGAP_{tot}]$. 
 In fact, if $C_2 \geq C_1$ then $\widehat{[mG^*]}=0$ from the same argument as in Case 3.
 Now, Eq. \ref{ss_red12:6} gives $\widehat{[tGEF^*]} = C_1 - \widehat{[mGAP^*]}$ and from Eq. \ref{ss_red12:5}, we must have  $\widehat{[mGAP^*]}= [mGAP_{tot}]$ or $\widehat{[tGEF^*]}=0$. 
 If $\widehat{[tGEF^*]}=0$ then $\widehat{[mGAP^*]} = C_1 \geq [mGAP_{tot}]$ and hence $\widehat{[mGAP^*]} = [mGAP_{tot}]$.  
Therefore, we have shown that Eq. \ref{cond4_12} imply $\widehat{[mG^*]}=0$ and $\widehat{[mGAP^*]}=[mGAP_{tot}]$. Consequently, the steady state in this case must be given by Eq. \ref{ss1_12_new}.
\\
\\
\noindent
\underline{Case 2:  $ \widehat{[tGEF]}=0$ and  $\widehat{[mGAP^*]}=[mGAP_{tot}]$ }
\noindent
\\
From Eq. \ref{ss_red12:7}, $\widehat{[tGEF^*]} = C_2 - [mGAP_{tot}] $ and hence $[mGAP_{tot}] \leq C_2$. 
From Eq. \ref{ss_red12:6}, we must have $ \widehat{[mG]} + \widehat{[mG^*]}  = C_1 - C_2$ and that implies $C_1 \geq C_2$. Now, Eq. \ref{ss_red12:1} gives 
$$ [mGEF^*] \left( C_1 - C_2 - \widehat{[mG^*]}  \right) = k [mGAP_{tot}] \widehat{[mG^*]}  $$
and therefore

$$ \widehat{[mG^*]} = \frac{[mGEF^*] \left( C_1 - C_2 \right)}{[mGEF^*] + k [mGAP_{tot}] } \quad \text{and} \quad  \widehat{[mG]} = \frac{k [mGAP_{tot}]\left( C_1 - C_2 \right)}{[mGEF^*] + k [mGAP_{tot}]}.$$

From Eq. \ref{ss_red12:3}, we must have 
$$  \left (C_2 - [mGAP_{tot}] \right)(1- \widehat{\mathcal{T}^*})  -  k [tGAP^*] \widehat{\mathcal{T}^*}  = 0 $$
from which we get
$$ \widehat{\mathcal{T}^*} = \frac{C_2 - [mGAP_{tot}]}{ \left(C_2 - [mGAP_{tot}]\right) + k [tGAP^*]}$$
and therefore the steady state is given by

\footnotesize
\begin{eqnarray}
\widehat{\bf{x}}
&=& \left(\frac{k [mGAP_{tot}]\left( C_1 - C_2 \right)}{[mGEF^*] + k [mGAP_{tot}]}, \frac{[mGEF^*] \left( C_1 - C_2 \right)}{[mGEF^*] + k [mGAP_{tot}] }, \right. \\%\label{ss2_12} \\
&& \left. \frac{C_2 - [mGAP_{tot}]}{ \left(C_2 - [mGAP_{tot}]\right) + k [tGAP^*]} ,0,C_2 - [mGAP_{tot}], [mGAP_{tot}] \right).
\nonumber
\end{eqnarray}
\normalsize

%%%%%%%%%%%%%%%%%%% Sufficient conditions SS case 2
\normalsize
We now observe that the two parameter relations
\begin{equation}
C_2 \geq [mGAP_{tot}] \quad \text{and} \quad C_1 \geq C_2
\label{cond2_12}
\end{equation}
 are sufficient for $ \widehat{[tGEF]}=0$ and $\widehat{[mGAP^*]}=[mGAP_{tot}]$. 
 
 In fact, if $C_1 \geq C_2$ then $\widehat{[tGEF]}=0$ from the same argument as in Case 1. 
 Now, Eq. \ref{ss_red12:7} gives $\widehat{[tGEF^*]} = C_2 - \widehat{[mGAP^*]}$ and from Eq. \ref{ss_red12:5}, we must have $ \widehat{[mGAP^*]} = [mGAP_{tot}]$ or $\widehat{[tGEF^*]}=0$. 
 If $\widehat{[tGEF^*]}=0$ then $  \widehat{[mGAP^*]} = C_2  \geq [mGAP_{tot}]$ (from Eq. \ref{cond2_12}) and thus $\widehat{[mGAP^*]} = [mGAP_{tot}]$.
 Therefore, we have shown that Eq. \ref{cond2_12} imply $ \widehat{[tGEF]}=0$ and $\widehat{[mGAP^*]}=[mGAP_{tot}]$. 
 Consequently, the steady state in this case must be given by Eq. \ref{ss2_12}.

\noindent
\underline{Case 3:  $\widehat{[mG^*]}=0$ and $\widehat{[tGEF^*]}=0$.}
\noindent
\\
From Eq. \ref{ss_red12:1}, we have $\widehat{[mG]} = 0$ and from Eq. \ref{ss_red12:3}, we also get $\widehat{\mathcal{T}^*} =0$ since $k$ and $[tGAP^*]$ are strictly positive numbers.
Now, Eq. \ref{ss_red12:6} gives $\widehat{[mGAP^*]} = C_1$ and thus we must have  $C_1 \leq [mGAP_{tot}]$.
Moreover, Eq. \ref{ss_red12:7} results in $\widehat{[tGEF]} = C_2 - C_1$ and since all steady states must be nonnegative, we obtain $C_2 \geq C_1$. In this case, the steady state is given by 
\small
\begin{eqnarray}
\widehat{\bf{x}}
&=& \left(0,0,0,C_2-C_1,0,C_1\right)
%\label{ss3_12}
\end{eqnarray}

%%%%%%%%%%%%%%%%%%% Sufficient conditions SS case 3
\normalsize
We now observe that the two parameter relations
\begin{equation}
C_1 \leq [mGAP_{tot}] \quad \text{and} \quad C_2 \geq C_1
\label{cond3_12}
\end{equation}
 are sufficient for $\widehat{[mG^*]}=0$ and $\widehat{[tGEF^*]}=0$. 
 In fact, by subtracting \ref{ss_red12:6} from Eq. \ref{ss_red12:7}, we obtain $$\widehat{[tGEF]} - \widehat{[mG]} + \widehat{[mG^*]} = C_2 - C_1 \geq 0 $$
 and hence $ \widehat{[tGEF]} \geq \widehat{[mG]} + \widehat{[mG^*]} $. On the other hand, from Eq. \ref{ss_red12:4}, we must have $ \widehat{[tGEF]}=0$ or $\widehat{[mG^*]}=0$. 
 Thus if $ \widehat{[tGEF]}=0$ then  $\widehat{[mG]} + \widehat{[mG^*]} \leq 0$ and hence the nonnegativeness of the steady state implies $\widehat{[mG]}=\widehat{[mG^*]}=0$. 
 Hence we conclude that $C_2 \geq C_1$ implies $\widehat{[mG^*]}=0$. 
 
 Now, Eq. \ref{ss_red12:6} gives $\widehat{[tGEF^*]} = C_1 - \widehat{[mGAP^*]}$ and from Eq. \ref{ss_red12:5}, we must have  $\widehat{[tGEF^*]}=0$ or $\widehat{[mGAP^*]}= [mGAP_{tot}]$. 
 If $\widehat{[mGAP^*]}=[mGAP_{tot}]$, then $\widehat{[tGEF^*]} = C_1 - [mGAP_{tot}]  \leq 0$ (from Eq. \ref{cond3_12}) and thus $\widehat{[tGEF^*]} = 0$. 
 Therefore, we have shown that Eq. \ref{cond3_12} imply $\widehat{[mG^*]}=0$ and $\widehat{[tGEF^*]} = 0$.
 Consequently, the steady state in this case must be given by Eq. \ref{ss3_12}.

\noindent
\underline{Case 4:  $ \widehat{[tGEF]}=0$ and $\widehat{[tGEF^*]} = 0$ }
\noindent
\\
From Eq. \ref{ss_red12:7}, we obtain $\widehat{[mGAP^*]}=C_2$ and hence $C_2 \leq [mGAP_{tot}]$. 
From Eq. \ref{ss_red12:6}, we have $\widehat{[mG]} + \widehat{[mG^*]} = C_1 - C_2$ and that implies $C_1 \geq C_2$ since the concentrations at steady state must be nonnegative. 
Eq. \ref{ss_red12:1} then gives
$$   - [mGEF^*]\left( C_1 -C_2 - \widehat{[mG^*]} \right) + k C_2\widehat{[mG^*]} = 0$$
from which we obtain 
$$ \widehat{[mG^*]} = \frac{[mGEF^*] \left( C_1 - C_2 \right)}{[mGEF^*] + k C_2} \quad \text{and} \quad  \widehat{[mG]} = \frac{k C_2 \left( C_1 - C_2 \right)}{[mGEF^*] + k C_2}.$$

From Eq. \ref{ss_red12:3}, we have $\widehat{\mathcal{T}^*}  = 0$ and therefore the steady state is given by

\footnotesize
\begin{eqnarray*}
\widehat{\bf{x}}
&=& \left(\frac{k C_2}{[mGEF^*]+k C_2} \left( C_1- C_2\right), \frac{[mGEF^*]}{[mGEF^*]+k C_2} \left( C_1- C_2\right), 0,0,0, C_2 \right). \nonumber
%\label{ss4_12_new}
\end{eqnarray*}

%%%%%%%%%%%%%%%%%%% Sufficient conditions old SS case 1 ---> new SS 4
\normalsize
We now observe that the two parameter relations
\begin{equation}
C_2 \leq [mGAP_{tot}] \quad \text{and} \quad C_1 \geq C_2
\label{cond1_12}
\end{equation}
 are sufficient for $ \widehat{[tGEF]}=0$ and $\widehat{[tGEF^*]} = 0$. 
 In fact, if $C_1 \geq C_2$ then by subtracting Eq. \ref{ss_red12:7} from Eq. \ref{ss_red12:6}, we have $$\widehat{[mG]} + \widehat{[mG^*]} - \widehat{[tGEF]}  = C_1 - C_2 \geq 0 $$
 and hence $ \widehat{[mG]} + \widehat{[mG^*]} \geq  \widehat{[tGEF]} $. On the other hand, from Eq. \ref{ss_red12:4}, we must have $ \widehat{[tGEF]}=0$ or $\widehat{[mG^*]}=0$. 
 Thus if $\widehat{[mG^*]} = 0$ then  $\widehat{[mG]} = 0$ (from Eq. \ref{ss_red12:1}) and hence the nonnegativeness implies $\widehat{[tGEF]}=0$. 
 Hence we conclude that Eq. \ref{cond1_12} guarantee $\widehat{[tGEF]}=0$.

Now, Eq. \ref{ss_red12:7} gives $\widehat{[tGEF^*]} = C_2 - \widehat{[mGAP^*]}$ and from Eq. \ref{ss_red12:5}, we must have $\left( [mGAP_{tot}] - \widehat{[mGAP^*]}  \right)=0$ or $\widehat{[tGEF^*]}=0$. 
If $\widehat{[mGAP^*]}=[mGAP_{tot}]$ then $\widehat{[tGEF^*]} = C_2 - [mGAP_{tot}]  \leq 0$ (from Eq. \ref{cond1_12}) and thus $\widehat{[tGEF^*]} = 0$. 
Therefore, we have shown that Eq. \ref{cond1_12} implies $\widehat{[tGEF]}=0$ and $\widehat{[tGEF^*]} = 0$.
Consequently, the steady state in this case must be given by Eq. \ref{ss4_12_new}.

\subsection*{Local Stability Analysis.}

We begin reducing the ODE system with the conservation laws given by Eqs. \ref{constrain1_12} and \ref{constrain2_12}. 
In fact, if we write 
$$[mG]  = C_1 -  [mG^*] - [tGEF^*] - [mGAP^*] \quad \text{and} \quad [tGEF] = C_2 - [tGEF^*] - [mGAP^*] $$

then Eqs. \ref{model_red12:1} -- \ref{model_red12:6} can be written in the form 

\small
\begin{eqnarray}
	\frac{d[mG^*]}{dt}       &=& f_1([mG^*],\mathcal{T}^*, [tGEF^*],[mGAP^*])    \\
	\frac{d\mathcal{T}^*}{dt}&=& f_2([mG^*],\mathcal{T}^*, [tGEF^*],[mGAP^*])  \\
	\frac{d[tGEF^*]}{dt}     &=& f_3([mG^*],\mathcal{T}^*, [tGEF^*],[mGAP^*]) \\
	\frac{d[tGEF^*]}{dt}     &=& f_4([mG^*],\mathcal{T}^*, [tGEF^*],[mGAP^*]) 
\end{eqnarray}
\normalsize

where 
\small
\begin{equation*}
\begin{aligned}
     f_1([mG^*],\mathcal{T}^*, [tGEF^*],[mGAP^*]) & = k_{on} [mGEF^*]\left( C_1 -  [mG^*] - [tGEF^*] - [mGAP^*] \right) \\ 
     &- k_{off}[mGAP^*][mG^*] -k_{on}\left( C_2 - [tGEF^*] - [mGAP^*]\right)[mG^*],
\end{aligned}
\end{equation*}

$$ f_2([mG^*],\mathcal{T}^*, [tGEF^*],[mGAP^*])  =  k_{on} [tGEF^*] (1 - \mathcal{T}^*) - k_{off}[tGAP^*]\mathcal{T}^*,$$

\begin{equation*}
\begin{aligned}
f_3([mG^*],\mathcal{T}^*, [tGEF^*],[mGAP^*]) & = k_{on}\left(C_2 - [tGEF^*] - [mGAP^*]\right)[mG^*]  \\
& - k_{on} [tGEF^*] \left( [mGAP_{tot}] - [mGAP^*] \right),
\end{aligned}
\end{equation*} 

and
$$f_4([mG^*],\mathcal{T}^*, [tGEF^*],[mGAP^*]) = k_{on}  \left( [mGAP_{tot}] - [mGAP^*] \right) [tGEF^*].$$
\normalsize

The eigenvalues of the Jacobian Matrix can be thus calculated for each one of the four steady states given by Eqs. \ref{ss1_12_new} -- \ref{ss4_12_new}.
We prove that all steady states are LAS by showing that the eigenvalues of the Jacobian Matrix are all negative real numbers. 
We perform the calculations with MATLAB's R2019b symbolic toolbox and analyze each case separately (see supplementary file with MATLAB codes). 
In what follows, let $k_{mGEF} = k_{on} [mGEF^*]$ and $k_{tGAP} = k_{off} [tGAP^*]$.

\begin{enumerate}
   
        \item If $C_1 > [mGAP_{tot}]$ and $C_2>C_1$, the  Jacobian matrix evaluated at the steady state given by Eq. \ref{ss1_12_new} gives the  eigenvalues
    
    $$ \lambda_1 =  -k_{on}(C_1 - [mGAP_{tot}])  \quad \text{and} \quad    \lambda_2 =  - k_{on} ( C_1 - [mGAP_{tot}] ) - k_{tGAP}  $$
    
    which are negative. Moreover, the other eigenvalues $\lambda_3$ and $\lambda_4$ are such that 
    $$  \lambda_3 + \lambda_4 =  k_{on} (C_1 - C_2) - k_{mGEF}  - k_{off} [mGAP_{tot}] < 0 $$  and 
  $$  \lambda_3 \lambda_4  = -k_{mGEF} k_{on}(C_1 - C_2) >0 $$
    
   and thus $\lambda_3$ and $\lambda_4$ are negative and hence the steady state is LAS.
    \item If $C_2 > [mGAP_{tot}]$ and $C_1>C_2$, the  Jacobian matrix evaluated at the steady state given by Eq. \ref{ss2_12} gives the  eigenvalues $$\lambda_1 =   - k_{mGEF} - k_{off}[mGAP_{tot}], \quad \lambda_2 = - k_{tGAP} - k_{on} (C_2 - [mGAP_{tot}]),$$ 
    $$\quad \lambda_3 = -k_{on} (C_2 - [mGAP_{tot}]) \quad \text{and} \quad  \lambda_4 = -\frac{k_{on} k_{mGEF} (C_1 - C_2)}{k_{off} [mGAP_{tot}] + k_{on} [mGEF^*]}$$

    which are all negative and hence the steady state is LAS.

        \item If $C_1 < [mGAP_{tot}]$ and  $C_2>C_1$,the  Jacobian matrix evaluated at the steady state given by Eq. \ref{ss3_12} gives the  eigenvalues

    $$ \lambda_1 =  k_{on} (C_1 - [mGAP_{tot}]) \quad \text{and} \quad    \lambda_2 = -k_{tGAP} $$
    
    which are negative. Moreover, the other eigenvalues $\lambda_3$ and $\lambda_4$ are such that 
    $$  \lambda_3 + \lambda_4 = k_{on}(C_1 -C_2)  - C_1 k_{off} - k_{mGEF} < 0 $$  and 
  $$  \lambda_3 \lambda_4  = -k_{mGEF}k_{on} (C_1 - C_2)>0 $$
    
   and thus $\lambda_3$ and $\lambda_4$ are negative and hence the steady state is LAS.

    \item If $C_2 < [mGAP_{tot}]$ and $C_1 > C_2$,  the Jacobian matrix evaluated at the steady state given by Eq. \ref{ss4_12_new}   gives the  eigenvalues 
        $$\lambda_1 = - k_{mGEF} - C_2 k_{off} , \quad  \lambda_2 = k_{on} (C_2 - [mGAP_{tot}]), \quad \lambda_3 = -k_{tGAP} $$ and 
 $$\lambda_4 = -\frac{k_{on} k_{mGEF}(C_1 - C_2)}{C_2 k_{off} + k_{on} [mGEF^*]}$$
which are all negative and hence the steady state is LAS.
\end{enumerate}

%%%%%%%%%%%%%%%%%%%%%%%%% PROOF OF THEOREM ARROWS 123 %%%%%%%%%%%%%%%%%%%%%%%%%%%%%%%
\newpage
\section{Proof of Theorem \ref{prop_123}}
\label{appendix_123}
\normalsize

We proceed with the steady state analysis in the same way of Theorem \ref{prop_12}.
We consider the same four different cases and calculate the $\xi$-dependent families of steady states, where $\xi \geq 0$ represent the tG concentration. We also obtain necessary relationships for the conserved quantities $\tilde{C_2}$, $\tilde{C_1}$, and $[mGAP_{tot}]$, as well as admissible intervals for $\xi$ that guarantee the existence of nonnegative steady states.

\underline{Case 1: $\widehat{[mG^*]}=0$ and $\widehat{[mGAP^*]}=[mGAP_{tot}]$ }. %%% old SS 4, new SS1

From Eq. \ref{ss_red123:1}, we have $\widehat{[mG]}=0$ and subtracting  Eq. \ref{constrain2_123} from Eq. \ref{constrain1_123}, we get $\widehat{[tGEF]} = \tilde{C_2} - \tilde{C_1} \geq 0$ only if $\tilde{C_2} \geq \tilde{C_1} $. Substituting  $\widehat{[tGEF]}$ on the conservation law given by Eq. \ref{constrain2_123} and using Eq. \ref{ss_red123:2} to write $\widehat{[tG^*]} = \xi \frac{[\widehat{tGEF^*]}}{k [tGAP^*]}$, we obtain 
$$ \xi + \xi \frac{\widehat{[tGEF^*]}}{k [tGAP^*]} + (\tilde{C_2} - \tilde{C_1}) + [mGAP_{tot}] = \tilde{C_2} $$

and hence 
 $$ \widehat{[tGEF^*]} = \left(\tilde{C_1} - [mGAP_{tot}] -\xi] \right) \frac{k [tGAP^*]}{k[tGAP^*] +\xi} $$
 only if $\tilde{C_1} -  [mGAP_{tot}] \geq \xi. $ Therefore, in this case the $\xi$-dependent family of steady states is given by 
\begin{eqnarray*}
\widehat{\textbf{x}}_{\xi}
&=& \Bigg(0, 0, \xi, \frac{\left(\tilde{C_1} - [mGAP_{tot}] -\xi] \right) \xi}{k[tGAP^*] +\xi}, \tilde{C_2} - \tilde{C_1}, \\\\\nonumber
&&
\left(\tilde{C_1} - [mGAP_{tot}] -\xi] \right) \frac{k [tGAP^*]}{k[tGAP^*] +\xi}, [mGAP_{tot}] \Bigg)  \nonumber
\end{eqnarray*}

\underline{Case 2: $ \widehat{[tGEF]}=0$ and  $\widehat{[mGAP^*]}=[mGAP_{tot}]$}

Using Eq. \ref{ss_red123:1} to write $\widehat{[mG]} = \frac{ k (\tilde{C_2} - \xi) }{[mGEF^*]} \widehat{[mG^*]} $ and subtracting  Eq. \ref{constrain2_123} from  Eq. \ref{constrain1_123}, we obtain the expressions for $[mG^*]$ and $[mG]$  
$$\widehat{[mG^*]} = \frac{(\tilde{C_1} - \tilde{C_2}) [mGEF^*]}{ k [mGAP_{tot}] + [mGEF^*]} \quad \text{and} \quad \widehat{[mG]} = \frac{(\tilde{C_1} - \tilde{C_2}) k(\tilde{C_2} - \xi) }{ k[mGAP_{tot}] + [mGEF^*]}$$
and clearly we must have $\tilde{C_1} \geq \tilde{C_2}$. 
Now looking at Eq. \ref{constrain2_123} and substituting $\widehat{[tG^*]} = \frac{\widehat{[tGEF^*]} \xi}{k [tGAP^*]}$, we obtain 
$$ \widehat{[tGEF^*]} = \left(\tilde{C_2} - [mGAP_{tot}] -\xi] \right) \frac{k [tGAP^*]}{k[tGAP^*] +\xi} $$
only if $\tilde{C_2} -  [mGAP_{tot}] \geq \xi$. 
Therefore, in this case the $\xi$-dependent family of steady states is given by 
\begin{eqnarray*}
\widehat{\textbf{x}}_{\xi}
&=& \Bigg(\frac{(\tilde{C_1} - \tilde{C_2}) k[mGAP_{tot}] }{ [mGEF^*] + k[mGAP_{tot}] }, \frac{(\tilde{C_1} - \tilde{C_2}) [mGEF^*]}{ [mGEF^*] + k[mGAP_{tot}]}, \\\\ \nonumber
&& \xi , \frac{\left(\tilde{C_2} - [mGAP_{tot}] -\xi] \right) \xi}{k[tGAP^*] +\xi} ,0, \left(\tilde{C_2} - [mGAP_{tot}] -\xi] \right) \frac{k [tGAP^*]}{k[tGAP^*] +\xi} , [mGAP_{tot}]\Bigg). \nonumber
\end{eqnarray*}

\underline{Case 3: $\widehat{[mG^*]}=0$ and $\widehat{[tGEF^*]}=0$ }

From Eqs. \ref{ss_red123:1} and \ref{ss_red123:2}, we have $\widehat{[mG]}=0$  and $\widehat{[tG^*]}=0$, respectively. Subtracting  Eq. \ref{constrain2_123} from Eq. \ref{constrain1_123}, in this case we get $\widehat{[tGEF]} = \tilde{C_2} - \tilde{C_1} \geq 0$ only if $\tilde{C_2} \geq \tilde{C_1} $. Now, from the conservation law given by Eq. \ref{constrain1_123}, we obtain $\widehat{[mGAP^*]} = \tilde{C_1} - \xi$ and $\widehat{[mGAP^*]} \in [0,[mGAP_{tot}]]$ only if $\max(0,\tilde{C_1} - [mGAP_{tot}]) \leq \xi \leq \tilde{C_1}$. In this case, the $\xi$-dependent family of steady states is given by 
\begin{eqnarray*}
\widehat{\textbf{x}}_{\xi}
&=& \left(0,0,\xi,0,\tilde{C_2}-\tilde{C_1},0,\tilde{C_1} - \xi \right).
\end{eqnarray*}

\underline{Case 4: $\widehat{[tGEF]}=0$ and $\widehat{[tGEF^*]} = 0$ } % Old SS1, new SS 4

Eq. \ref{ss_red123:2} gives $\widehat{[tG^*]}=0$ and the conservation law given by Eq. \ref{constrain2_123} yields  
$[mGAP^*]  = \tilde{C_2} - \xi$. 
Now using Eq. \ref{ss_red123:1} to write $\widehat{[mG]} = \frac{ k (\tilde{C_2} - \xi)}{[mGEF^*]} \widehat{[mG^*]} $, the conservation law given by Eq. \ref{constrain2_123} gives  
$$\widehat{[mG^*]} = \frac{(\tilde{C_1} - \tilde{C_2}) [mGEF^*]}{ k(\tilde{C_2} - \xi) + [mGEF^*]} \quad \text{and} \quad \widehat{[mG]} = \frac{(\tilde{C_1} - \tilde{C_2}) k(\tilde{C_2} - \xi) }{ k(\tilde{C_2} - \xi) + [mGEF^*]} $$
and since $[mGAP^*] \in \left[0,[mGAP_{tot}]\right]$ and the steady states must be nonnegative,
we must have $$\max(0,\tilde{C_2} - [mGAP_{tot}]) \leq \xi \leq \tilde{C_2} \leq \tilde{C_1}.$$

The $\xi$-dependent familiy of steady states is therefore given by 
\begin{eqnarray*}
\widehat{\textbf{x}}_{\xi}
&=& \left(\frac{(\tilde{C_1} - \tilde{C_2}) k(\tilde{C_2} - \xi) }{ [mGEF^*] + k(\tilde{C_2} - \xi) }, \frac{(\tilde{C_1} - \tilde{C_2}) [mGEF^*]}{ [mGEF^*] + k(\tilde{C_2} - \xi)}, \xi ,0,0,0, \tilde{C_2}- \xi \right). \nonumber
\end{eqnarray*}

%\end{comment}

\end{document}